\documentclass[12pt]{article}
\pdfoutput=1
\usepackage[T1]{fontenc}
\usepackage{pdfsync}
\usepackage[utf8]{inputenc}
\usepackage[english]{babel}
\usepackage{nomencl}
\usepackage{eqnarray}
\usepackage{glossaries}
\usepackage{tocloft}
\usepackage{stmaryrd}
\usepackage{color}  
\usepackage{pgfplots}
\pgfplotsset{compat=newest}
\usepackage{cite}

\usepackage{mathabx, eufrak}

\usepackage{amscd,amsmath}
\usepackage{mathrsfs}
\usepackage{geometry}
\usepackage{lipsum}
\usepackage{imakeidx}
\usepackage{physics}

\usepackage{tikz-feynman}
\usepackage[pre-1.03]{feynmp}
\tikzfeynmanset{compat=1.1.0}

\usepackage{bbold}
\usepackage{bm}

\usepackage{cancel}

\usepackage{amsthm}

\usepackage{authblk}

\usepackage{hyperref}
\hypersetup{
    colorlinks=true,
    citecolor=blue,    
    linkcolor=red,     
    urlcolor=magenta   
}

\theoremstyle{definition}

\theoremstyle{remark}
\newtheorem*{remark}{Remark}

\theoremstyle{plain}
\newtheorem*{lemma}{Lemma}
\newtheorem*{theorem}{Theorem}
\newtheorem*{corollary}{Corollary}

\def\be#1\ee{\begin{align}#1\end{align}}
\newcommand{\la}{\label}
\newcommand{\van}{\scriptscriptstyle}

\def\d{\delta}

\newcommand{\bold}{\boldsymbol}

\newcommand{\ie}{\emph{i.e.} }
\newcommand{\eg}{\emph{e.g.} }

\newcommand{\cS}{\mathcal{S}}
\newcommand{\cM}{\mathcal{M}}

\newcommand{\cF}{\mathcal{F}}

\newcommand{\cH}{\mathcal{H}}
\newcommand{\cO}{\mathcal{O}}
\newcommand{\cL}{\mathcal{L}}

\newcommand{\cK}{\mathcal{K}}
\newcommand{\cQ}{\mathcal{Q}}
\newcommand{\cP}{\mathcal{P}}

\newcommand{\teps}{\tilde{\bold{\epsilon}}}
\newcommand{\eps}{{\bold{\epsilon}}}
\newcommand{\beps}{\bar{\bold{\epsilon}}}
\def\heq{\,\hat{=}\,}

\newcommand{\f}{\frac}

\newcommand{\lbr}{\llbracket}
\newcommand{\rbr}{\rrbracket}
\def\EH{\mathrm{EH}}

\newcommand{\pa}{\mathop{}\!\partial}

\makeatletter

\title{\Large\textbf{
On the geometric origin of central charges in gravity
}}

\author[1,2,3]{Lorenzo Di Giacomo
\footnote{\href{mailto:digiacomo.lorenzo@spes.uniud.it}{digiacomo.lorenzo@spes.uniud.it} }
}

\author[1,2]{Daniele Pranzetti
\footnote{\href{mailto:daniele.pranzetti@uniud.it}{daniele.pranzetti@uniud.it} }
}
\date{}

\affil[1]{\small\emph{Università degli Studi di Udine, via Palladio 8, I-33100 Udine, Italy}}
\affil[2]{\small\emph{Institute for Fundamental Physics of the Universe (IFPU), Via Beirut 2, 34151 Trieste, Italy}}
\affil[3]{\small\emph{Istituto Nazionale di Fisica Nucleare, sezione di Trieste, Italy}}

\begin{document}
\newgeometry{right=2.5cm, left=2.5cm}
\maketitle

\vspace{-1cm}
\begin{abstract}

Central extensions of asymptotic symmetry algebras are traditionally determined through an explicit computation of the algebra of conserved charges. We show that, within an ambiguity-free covariant phase space formalism based on the Noetherian split, they instead arise as intrinsic quantities completely determined by the Lagrangian through the geometry of the associated covariant phase space. We prove two general Charge Representation Theorems describing the canonical realization of diffeomorphism symmetries in the presence of field-dependent vector fields and Lagrangian anomalies, corresponding respectively to vanishing and integrable Noetherian flux. In both cases, we derive explicit expressions for the associated central extensions and establish the corresponding cocycle identities. This provides a constructive algorithm for determining central charges directly from the Lagrangian formulation of the theory, without requiring an explicit computation of the charge algebra.
As a non-trivial application, we consider asymptotically locally AdS$_3$ spacetimes with Penrose--Brown--Henneaux boundary conditions. We recover the corresponding centrally extended asymptotic symmetry algebra, providing an explicit realization of  the Charge Representation Theorems.

\end{abstract}

\newpage
\tableofcontents

\section{Introduction}

The construction of conserved charges associated with gauge symmetries is one of the central problems in classical and quantum gravity. Besides providing the canonical generators of gauge transformations, conserved charges encode fundamental physical information about gravitational systems, ranging from black hole mechanics and asymptotic symmetries to holography and the microscopic interpretation of gravitational entropy. In particular, the appearance of central extensions in the representation of asymptotic symmetry algebras has profound physical consequences, the most celebrated example being the Brown--Henneaux central charge governing the asymptotic symmetry algebra of three-dimensional anti-de Sitter gravity \cite{Brown:1986nw}. This result established the infinite-dimensional nature of the asymptotic symmetry group of AdS$_3$ gravity and laid the foundations for the AdS$_3$/CFT$_2$ correspondence.

The modern understanding of conserved charges is largely based on the covariant phase space formalism initiated by Ashtekar et al. \cite{Ashtekar:1981bq,Ashtekar:1990gc} and  further developed by Lee and Wald \cite{Lee:1990nz} and   Iyer and Wald \cite{Iyer:1994ys}. In parallel, a complementary cohomological approach was developed by Barnich and Brandt  \cite{Barnich:2001jy,Barnich:2003xg}, leading to equivalent expressions for conserved currents and asymptotic charges in a broad class of gravity and gauge theories. These constructions have become the standard framework for the analysis of gravitational symmetries and have found numerous applications in black hole thermodynamics, asymptotically flat and anti-de Sitter spacetimes, soft theorems and holography. A comprehensive review can be found in \cite{Compere:2018aar}.

Despite their remarkable success, several conceptual issues remain only partially understood. One of the most persistent concerns the ambiguity in the definition of the pre-symplectic potential. Given a Lagrangian, the latter is only defined up to the addition of an exact spacetime form, and different representatives generally lead to different expressions for charges and symplectic fluxes \cite{Iyer:1994ys,Compere:2018aar}. Although these ambiguities leave the equations of motion unaffected, they obscure the identification of canonical Hamiltonian generators and, consequently, of the corresponding symmetry algebra and its possible central extensions.

Additional subtleties arise whenever the symmetry generators become field-dependent. This situation naturally occurs after imposing gauge conditions or asymptotic boundary conditions and requires replacing the ordinary Lie bracket by a modified bracket taking into account the field dependence of the vector fields \cite{Barnich:2009se,Barnich:2010eb, Barnich:2011mi}. At the same time, the presence of a {non-integrable} symplectic flux through the boundary generally prevents the existence of Hamiltonian generators, so that the canonical Poisson bracket {representing the symmetry action and algebra} ceases to exist. Different prescriptions have been proposed to address this problem, including the Wald--Zoupas construction \cite{Wald:1999wa} and the Barnich--Troessaert bracket \cite{Barnich:2010eb,Barnich:2011mi}. More recently, these issues have been revisited from the perspective of extended phase spaces, edge modes, corner and asymptotic symmetries \cite{Donnelly:2016auv,Speranza:2017gxd, Geiller:2017xad, Geiller:2017whh, Geiller:2019bti,Freidel:2020xyx,Freidel:2020svx,Freidel:2020ayo,Freidel:2021cjp,Freidel:2021fxf, Chandrasekaran:2020wwn,Ruzziconi:2020wrb,Margalef-Bentabol:2020teu,Ciambelli:2021nmv,Ciambelli:2021vnn,Freidel:2021dxw,Chandrasekaran:2021vyu,Carrozza:2021gju,Adami:2021nnf,Wieland:2021eth,Ciambelli:2022cfr,Odak:2022ndm,Carrozza:2022xut,Nagy:2022xxs,Kabel:2023jve,Klinger:2023qna,FrancoisAndre:2023jmj,Baulieu:2024oql,Adami:2024gdx,Rignon-Bret:2024wlu,Langenscheidt:2024nyw,Baulieu:2024rfp,Neri:2025fsh,Golshani:2026lvc}, leading to a deeper understanding of the geometric origin of boundary degrees of freedom.

Three-dimensional gravity with negative cosmological constant provides an ideal arena in which these conceptual questions can be investigated explicitly. Besides its role in holography, the theory possesses a rich asymptotic symmetry structure while remaining sufficiently simple to allow for complete analytic control. The Brown--Henneaux boundary conditions \cite{Brown:1986nw} constitute the prototypical example of asymptotically AdS$_3$ boundary conditions, leading to two copies of the Virasoro algebra with classical central charge. Within the Fefferman--Graham expansion \cite{Fefferman:2007rka}, the asymptotic geometry admits a natural holographic interpretation in terms of the boundary metric and the holographic stress tensor \cite{Balasubramanian:1999re,Skenderis:1999nb,Skenderis:2000in,deHaro:2000vlm}.

More recently, considerable attention has been devoted to generalized asymptotically AdS$_3$ boundary conditions. In particular, allowing the boundary metric to fluctuate within its conformal class leads to the so-called Penrose--Brown--Henneaux  boundary conditions, where Weyl transformations become part of the asymptotic symmetry group and Lagrangian anomalies play an essential role. These generalized boundary conditions have been investigated from several complementary perspectives, including holographic renormalization, asymptotic symmetries and conformal anomalies, revealing a richer canonical structure than in the original Brown--Henneaux analysis \cite{Compere:2013bya,Troessaert:2013fma,Grumiller:2016pqb,Grumiller:2017sjh,Henneaux:2019sjx,Alessio:2020ioh}. They therefore provide a particularly suitable laboratory for studying the interplay between anomalies, integrability and central extensions within the covariant phase space formalism.

The present work builds upon the ambiguity-free formulation of the covariant phase space  introduced in \cite{Freidel:2020xyx,Freidel:2021cjp}. The central ingredient of this framework is the use of Anderson's homotopy operators \cite{anderson1992introduction} to associate a unique pre-symplectic potential to a given Lagrangian. As a consequence, the fundamental canonical relation naturally decomposes into two geometrically distinguished quantities, namely a Noether charge and a Noetherian flux. This \emph{Noetherian split} eliminates the traditional ambiguities affecting the definition of charges and fluxes, leaving only genuine boundary Lagrangian shifts whose effects can be analyzed systematically.

The main objective of this paper is to investigate the canonical realization of diffeomorphism symmetries within this ambiguity-free framework in the simultaneous presence of field-dependent vector fields and Lagrangian anomalies. Our first result consists in proving two general Charge Representation Theorems. The first theorem applies whenever the Noetherian flux vanishes. In this case, the Noether charge itself provides a Hamiltonian realization of the symmetry algebra through the canonical Poisson bracket, and we derive an explicit expression for the associated field-independent central extension. The second theorem addresses the more general situation in which the Noetherian flux does not vanish but is nevertheless integrable. In this case the Hamiltonian generator is obtained by subtracting the integrable flux contribution from the Noether charge, and we derive the corresponding expression for the central extension.

These results  extend the scope of the ambiguity-free covariant phase space construction. In particular, they provide a constructive algorithm for determining central charges directly from the Lagrangian formulation of the theory. Once the unique pre-symplectic potential has been constructed and the Noetherian split has been identified, the determination of the central extension requires neither an explicit computation of the charge algebra nor the introduction of additional prescriptions. Instead, the central charge follows directly from the geometric structures defining the covariant phase space. In the case of an integrable Noetherian flux, we further prove that the resulting central extension is invariant under boundary Lagrangian shifts preserving the integrability properties of the theory, showing that it is an intrinsic property of the covariant phase space rather than an artefact of a particular boundary representative.  In this sense, the Noetherian split not only provides a natural separation between charges and fluxes, but also identifies the geometric mechanism responsible for the appearance of central charges in the representation of diffeomorphism symmetries.

As a non-trivial application of the general formalism, we consider three-dimensional Einstein gravity with negative cosmological constant subject to Penrose--Brown--Henneaux boundary conditions.  This framework naturally gives rise to Lagrangian anomalies and to a non-vanishing but integrable Noetherian flux, providing an ideal testing ground for the general results established in the first part of the paper.

Applying our formalism, we derive the Hamiltonian generators directly from the Einstein--Hilbert Lagrangian and we test our prescription for determining the central charge of the theory. More precisely, we compute first the corresponding central extension  explicitly evaluating the charge algebra and we then verify that this is reproduced exactly by the general expression obtained in the second Charge Representation Theorem. The theorem reproduces the central charge previously obtained in the literature for Penrose--Brown--Henneaux boundary conditions, thereby providing a first-principles derivation of the canonical representation of the asymptotic symmetry algebra. Ordinary Brown--Henneaux boundary conditions are then recovered as the special case in which the Noetherian flux vanishes, furnishing a direct realization of the first Charge Representation Theorem.

The paper is organized as follows. In Section~2 we review the ambiguity-free covariant phase space formalism and  the Noetherian split in the presence of field-dependent diffeomorphisms and Lagrangian anomalies. We then prove two Charge Representation Theorems, corresponding respectively to vanishing and integrable Noetherian flux, and derive explicit expressions for the associated central extensions. In Section~3 we apply the general construction to asymptotically locally AdS$_3$ spacetimes with Penrose--Brown--Henneaux boundary conditions, derive the Hamiltonian generators and reproduce the corresponding centrally extended asymptotic symmetry algebra. Technical derivations are collected in the appendices.

\paragraph{Notation and conventions:\label{NC}} Throughout this paper, we use the convention $8\pi G=1=c$.  The spacetime $\cM$ is equipped with a Lorentzian metric $g_{\mu\nu}$  with a signature ($-++$).   
Greek letters denote the 3-dimensional metric and Latin letters the induced 2-dimensional metric on the time-like cylinder. Cartan's calculus in spacetime consists of the following notation: $\iota_\xi$ for the  vector field $\xi$ contraction on spacetime forms, $d$ the exterior derivative and $\cL_\xi$ the Lie derivative. In field-space we have $I_\xi$ to denote the contraction on field-space forms, $\delta$ and $\delta_\xi$ for the derivative operators respectively. In particular, we use the boldface notation to denote the spacetime $p$-forms. The on-shell of the bulk constraints/equations of motion equality is denoted by the symbol $\heq$.

The volume-form is denoted by $\bold{\epsilon}= 
\sqrt{|g|} d^3 x$, with orientation $d^3 x=dt\wedge dr\wedge d\phi$. Upon a choice of foliation, we introduce a $2+1$ decomposition
in terms of a codimension-1 space-like slices $\Sigma$ such that $\cM=\Sigma\times\mathbb{R}$. Each slice is equipped with a normal 1-form $\bold{n}=n_\mu d x^\mu$, which  we normalize  by demanding that it satisfies $g^{\mu\nu}n_\mu n_\nu=-1$. The metric-dependent normal vector is denoted  by $\hat{n}=n^\mu\partial_\mu$.
The slice has an induced metric $\tilde{g}_{\mu\nu}:= g_{\mu\nu}+n_\mu n_\nu$. We denote the volume form on $\Sigma$ by $\teps$; this is related to the spacetime volume form by $\eps=\bold{n}\wedge\teps$, which means $\teps=-\iota_{\hat{n}} \eps$.
The volume form contraction  $\eps_{\mu}=\iota_{\pa_\mu}\eps $
denotes a basis of  codimension-1 forms,  
In this notation, the Stokes theorem reads
\be\la{Stokes}
\int_\cM \eps\nabla_\mu v^\mu=\int_\cM d(v^\mu\eps_\mu)=\int_{\pa \cM}\eps_\mu v^\mu=\int_{\Sigma}\teps\,n_\mu v^\mu,
\ee

Since we are interested in AdS$_3$ spacetime, we also consider a time-like boundary $\Gamma$, with space-like normal 1-form $\bold{s}=s_\mu dx^\mu$ and vector $\hat{s}=s^\mu\pa_\mu$;  the associated induced metric is denoted by $\bar{g}_{\mu\nu}:= g_{\mu\nu}-s_\mu s_\nu$. Its volume form is such that $\eps=-\bold{s} \wedge\beps$, with $\beps=-\iota_{\hat s} \eps=\sqrt{|\bar g|}\,dt\wedge d\phi$. In this case, the Stokes theorem \eqref{Stokes} contains a second contribution on the RHS at the time-like cylinder with $(\Sigma,\teps, n_\mu)$ replaced by $(\Gamma,\beps, s_\mu)$. The boundary $\Gamma$ intersect the space-like $\Sigma$ at a codimension-2 corner $S$ (a circle on the cylinder at
spatial infinity, in the present case), and its associated volume form is given by $\bar{\tilde{\eps}}=\iota_{\hat{s}}\iota_{\hat{n}}\eps=\sqrt{\bar{\tilde{g}}}\,d\phi$, with $\bar{\tilde{g}}_{\mu\nu}:=\tilde{g}_{\mu\nu}-s_\mu s_\nu$.

\section{Ambiguity-free covariant phase space}\label{sec:ECS}

In this section we develop the covariant phase space formalism that we apply later to the case of  asymptotically locally AdS$_3$  spacetimes. Our approach is based on the Noetherian split introduced in \cite{Freidel:2021cjp}. We begin our treatment by recalling the main elements of the covariant phase space formalism  and the construction of a unique pre-symplectic potential $\bold{\theta}=\bold{\theta}^L$ starting from a given Lagrangian $\bold{L}$, and defining the Noetherian split between charge and flux. We then specialize to the case of a vanishing Noetherian flux and of an integrable Noetherian flux; in both cases we prove the corresponding {\it Charge representation theorem}, giving the explicit expression of the central charge.

\subsection{Noetherian split }\label{sec:NP}

The often times emphasized ambiguities in  $\bold{\theta}$  \cite{Wald:1999wa}, can be completely removed by applying Anderson's homotopy operators \cite{anderson1992introduction,Compere:2018aar} to derive a pre-symplectic potential \cite{Lee:1990nz, Freidel:2020xyx} such that, given $\bold{L}$,  the relation 
\begin{equation} 
\delta \bold{L} = d\bold{\theta} -\bold{E}\,,
\la{dL}
\end{equation}
where $\bold{E}$ are the equations of motion, follows directly from its definition. 
 From the field-derivative $\d$ of the pre-symplectic potential, we can define the pre-symplectic 2-form $\Omega$ by integrating it on a codimension-1 Cauchy slice $\Sigma$ \begin{equation}
 \Omega:= \var\int_\Sigma \bold{\theta}=\var {\Theta}.
 \end{equation}
The knowledge of the pre-symplectic 2-form allows us to derive the symmetry
charges of the theory and their  algebra. However, application of the  second Noether theorem receives a first layer of complexity when the vector field symmetry generators $\xi$
are field-dependent, \ie $\delta\xi\neq0$, happening for instance when working in a fixed coordinate-gauge, as often times the case when studying asymptotic symmetries. 
In this case, one has to introduce a modified Lie bracket, denoted by the double square brackets $\lbr \cdot, \cdot \rbr $, which takes into account the field-dependence and   is defined in the following manner  \cite{Barnich:2010eb}
\begin{equation}  
\lbr \xi, \zeta \rbr := [\xi, \zeta]_{\text{Lie}} +\delta_\zeta \xi -\delta_\xi \zeta,
\label{mod_lie}
\end{equation}
and such that the commutator of two  symmetry actions
 $\delta_\xi = \delta I_\xi +I_\xi \delta$ satisfies
 \be
 [ \delta_\xi,  \delta_\chi]= -\delta_{ \lbr \xi,\chi\rbr}\,.
 \label{maincom}
 \ee

Moreover, in the presence of gauge-fixings and imposition of boundary conditions,
the actions of spacetime  Lie derivative ${\cL_\xi}= d \iota_\xi +\iota_\xi d$ and field-space Lie derivative $\delta_\xi $ no longer coincide in general, \ie  $\delta_\xi \neq \cL_\xi$. In order to express this discrepancy,
one introduces the field contraction $I_{\delta \xi}$ and defines the {\it anomaly operator} \cite{Chandrasekaran:2020wwn,Freidel:2021cjp,Freidel:2021fxf}
\begin{equation}
\Delta_\xi := \delta_\xi - \cL_\xi - I_{\delta \xi}\,.
\la{ano}
\end{equation}

We are usually interested in theories where the Lagrangian is semi-covariant under diffeomorphisms, \ie the action of the  anomaly operator on the Lagrangian consists of a purely boundary term
\be
\Delta_\xi \bold{L}=d\bold{a}_\xi\,,
\label{DL}
\ee
with $\bold{a}_\xi$ denoted the {\it Lagrangian anomaly}. These anomalies can appear when the theory Lagrangian does not vanish on-shell, like in the presence of a non-vanishing cosmological constant, and they play a central role in the analysis of asymptotically locally AdS$_3$ spacetimes that we perform in second part of this work.

The anomaly operator action \eqref{DL},   combined with \eqref{dL}, implies that the action of the  anomaly operator on the pre-symplectic potential  
can be written as
\begin{equation}
 \Delta_\xi \bold{\theta}=\delta \bold{a}_\xi -\bold{a}_{\var\xi}+d\bold{A}_\xi,
\end{equation}
where $\bold{A}_\xi$ is a field-space 1-form called the \emph{symplectic anomaly}.

In the presence of field-dependence and anomalies the on-shell
Noether charge and the Noetherian flux are  respectively given by  \cite{Freidel:2021cjp}
\begin{equation}
\cQ_\xi \heq \int_S \bold{q}_\xi\,, \qquad \text{where}\quad d\bold{q}_\xi \heq I_\xi\bold{\theta} - \iota_\xi \bold{L}-\bold{a}_\xi\,,
 \label{charge_exp}
\end{equation} 
and
\begin{equation}
\cF_\xi \heq   \cF_\xi^{\theta} + \int_S (\bold{q}_{\delta\xi}+\bold{A}_\xi)\,, \qquad \text{where}\quad \cF_\xi^{\theta} = \int_S \iota_\xi\bold{\theta}
\label{flux}
\end{equation}
is the {\it symplectic flux}, and $S=\pa\Sigma$ is a codimension-2 surface often referred to as the {\it corner}. These are such that the on-shell {\it fundamental canonical relation}
\begin{equation}
-I_\xi\Omega \heq \delta \cQ_\xi -\cF_\xi
\label{ham_ch}
\end{equation}
is satisfied.

If one shifts the original Lagrangian by a boundary Lagrangian $\bold{\ell}$,   one can again uniquely associate a pre-symplectic potential to the new Lagrangian $\bold{L}'= \bold{L}+d\bold{\ell}$, which is given by
\begin{equation}
 \bold{\theta}' = \bold{\theta}+ \delta\bold{\ell}  -d\bold{\vartheta}\,,
 \la{theta-shift}
 \end{equation}
where $\bold{\vartheta}$ is the corner pre-symplectic potential, again uniquely associated with the boundary Lagrangian by means of the proper Anderson’s homotopy operator. The new Lagrangian defines a new Noether charge and Noetherian flux related to the previous ones by the following relations \cite{Freidel:2021cjp}
\begin{equation}
\cQ_\xi' = \cQ_\xi + \int_S (\iota_\xi\bold{\ell} -I_\xi \bold{\vartheta})\,,
\qquad
\cF_\xi' = \cF_\xi +\int_S (\delta\iota_\xi\bold{\ell}-\delta_\xi\bold{\vartheta})
\label{mod_QF}\,.
\end{equation}
Due to the extra terms in the pre-symplectic potential \eqref{theta-shift}, the fundamental canonical relation \eqref{ham_ch} is preserved by these shifts. 
We now specialize to the cases where the Noetherian flux is either vanishing or integrable; these are of relevance for AdS$_3$ spacetimes, which provide an arena to test the charge representation theorems we prove in the following. 

\subsection{Vanishing Noetherian flux}\la{sec:F=0}

When the Noetherian flux  vanishes, \ie $\cF_\xi=0$,  the Noether charge is integrable and is the canonical generator of the symmetry, namely it is the Hamiltonian charge. Under this condition of integrability, we can define a canonical Poisson bracket 
 \begin{equation}
     \{\cQ_\xi, \cQ_\zeta\}_L := -I_\xi I_\zeta \Omega =\var_\xi\cQ_\zeta. \label{Poi}
 \end{equation}
In order to state the charge representation theorem, it is necessary to enunciate the following lemma \cite{Freidel:2021cjp}.
\begin{lemma}
    If the ($d$-1)$th$-homology group of a $d$-dimensional spacetime $\cM$ is trivial, i.e. $H_{d-1}(\cM)=0$, then there exists a codimension-2 form $\bold{c}_{(\xi,\zeta)}$ such that 
    \begin{equation}
        d\bold{c}_{(\xi,\zeta)} = \Delta_\xi\bold{a}_\zeta-\Delta_\zeta\bold{a}_\xi +\bold{a}_{\lbr \xi, \zeta \rbr }\,,
    \end{equation}
and which is left invariant by Lagrangian shifts.
\end{lemma}

\begin{theorem}[\textbf{Charge representation theorem 1}]
    For Noether charges with vanishing Noetherian flux, the algebra given by the Poisson bracket in \eqref{Poi} is isomorphic to the Lie algebra of diffeomorphisms \eqref{mod_lie} up to a central extension, 
    \begin{equation}
     \{\cQ_\xi, \cQ_\zeta\}_L \heq -\cQ_{\lbr \xi,\zeta \rbr } - \cK^L_{(\xi, \zeta)}\,,
     \label{ch_al}
    \end{equation}
    where 
    \begin{equation}
         \cK^L_{(\xi, \zeta)}=\int_{\cS}\iota_\xi\iota_\zeta \bold{L} + \int_\cS(\iota_\xi \bold{a}_\zeta-\iota_\zeta \bold{a}_\xi)+\int_\cS \bold{c}_{(\xi,\zeta)} 
         \label{cocy}
    \end{equation}
    {is a field-space constant, namely a central charge,}
    and the 2-cocycle identity holds
    \begin{equation}
        \var_\xi \cK^L_{(\zeta, \eta)}+ \cK^L_{(\lbr \xi, \zeta]\! ], \eta)} + \text{cycl.} =0. \label{id_co}
    \end{equation}
\end{theorem}
\vspace{0.05cm}
\begin{proof}
Using the charge expression in \eqref{charge_exp} and substituting into \eqref{Poi}, we obtain (see Appendix A.5 in \cite{Freidel:2021cjp} for details),  
\begin{equation}
    \begin{aligned}
        \var_\xi \cQ_\zeta &\heq    -\cQ_{\lbr \xi,\zeta\rbr } -\cK^L_{(\xi, \zeta)} + I_\zeta I_\xi \Omega + \var_\zeta \cQ_\xi.
    \end{aligned}
    \label{14}
\end{equation}
Using the fundamental canonical relation \eqref{ham_ch}, where $\cF_\xi=0$, we obtain
\begin{equation}
\var_\xi \cQ_\zeta \heq -\cQ_{\lbr \xi,\zeta\rbr } -\cK^L_{(\xi, \zeta)},
\la{dxiQ}
\end{equation}
from which  \eqref{ch_al} follows immediately. 

To prove that $\cK^L_{(\xi, \zeta)}$ is a central charge, namely a field-independent cocycle, we compute
\be
\delta \cK^L_{(\xi, \zeta)}&\heq
-\delta \var_\xi \cQ_\zeta
-\delta\cQ_{\lbr \xi,\zeta\rbr }
\cr
&=(\var_\xi I_\zeta +I_{\lbr \xi,\zeta\rbr })\Omega
\cr
&=I_\zeta \var_\xi  \Omega
\cr
&=I_\zeta \var I_\xi  \Omega
\cr
&=-I_\zeta \var^2 \cQ_\xi  
\cr
&=0\,,
\ee
where we used the commutator 
\be
[\var_\xi, I_\zeta] = -I_{\lbr \xi, \zeta \rbr }\,,
\la{dxiI}
\ee
 and the field differential property $\var^2=0$.

To prove the cocycle identity, we recall that the pre-symplectic 2-form is closed, then
\begin{equation}
    I_\xi I_\zeta I_\eta \var\Omega=0
\end{equation}
and using repeatedly the previous commutation relation \eqref{dxiI}, one can prove\footnote{We denote $\Omega(\var_\zeta, \var_\eta)=I_\eta I_\zeta \Omega$.}
\begin{equation}
    \var_\xi \Omega(\var_\zeta, \var_\eta) + \Omega(\var_{\lbr \xi, \zeta \rbr }, \var_\eta) + \text{cycl.}=0 \,,
    \label{id_om}
\end{equation}
where  cycl. refers to the cyclic permutation of $(\xi,\zeta,\eta)$. Next, rewriting  the identity \eqref{14} as
\begin{equation}
    \Omega(\var_\zeta, \var_\eta) \heq \var_\zeta \cQ_\eta - \var_\eta \cQ_\zeta + \cQ_{\lbr \zeta, \eta\rbr } +\cK^L_{(\zeta,\eta)} \,,
    \la{OQ}
\end{equation}
the 
 charge  terms in \eqref{OQ} can be shown to drop out of \eqref{id_om}, as
\begin{equation}
\begin{aligned}
    &\var_\xi \var_\zeta \cQ_\eta - \var_\xi\var_\eta \cQ_\zeta + \var_\xi\cQ_{\lbr \zeta, \eta\rbr } + \var_{\lbr \xi, \zeta \rbr } \cQ_\eta - \var_\eta \cQ_{\lbr \xi, \zeta \rbr } + \cQ_{\![ \lbr \xi, \zeta \rbr , \eta\rbr } +\text{cycl.}\\
    &= [\var_\xi, \var_\zeta]\cQ_\eta -[\var_\xi, \var_\eta]\cQ_\zeta + [\var_\zeta, \var_\eta]\cQ_\xi  + \var_{\lbr \xi, \zeta \rbr } \cQ_\eta - \var_{\lbr \xi, \eta\rbr } \cQ_\zeta+ \var_{\lbr \zeta, \eta\rbr } \cQ_\xi\\
    & + \cQ_{\lbr \lbr \xi, \zeta \rbr , \eta\rbr } + \cQ_{\lbr \lbr \eta, \xi\rbr , \zeta\rbr }+ \cQ_{\lbr \lbr \zeta, \eta\rbr , \xi\rbr }\\
    &=0,
\end{aligned}
\end{equation}
where we used the Jacobi identity of the bracket \eqref{mod_lie} and the commutation relation \eqref {maincom}.
Therefore, since the charge terms in \eqref{id_om} vanish, it follows that the 2-cocycle identity \eqref{id_co} is satisfied.
\end{proof}

As a consistency check, we now prove the following corollary.

\begin{corollary}
 Under a boundary Lagrangian shift preserving the vanishing of the Noetherian flux, \ie such that
\be
\delta\iota_\xi\bold{\ell}=\delta_\xi\bold{\vartheta}\,,
\la{ellvarphi}
\ee
the shifted cocycle $\cK^{L'}_{(\xi, \zeta)}$ is still a field-space constant, and the charge bracket relation \eqref{ch_al} is preserved, that is 
 \be
   \{Q'_\xi,Q'_\zeta\}_{L'}+\cK^{L'}_{(\xi,\zeta)} + Q'_{\lbr \xi,\zeta\rbr} = 
  \{Q_\xi,Q_\zeta\}_{L} +\cK^{L}_{(\xi,\zeta)} + Q_{\lbr \xi,\zeta\rbr} \,.
  \label{brack4}
 \ee
\end{corollary}

\begin{proof}
In order to prove the first part of the corollary, let us first recall that the general expression of the shifted cocycle is given by
\begin{equation}
\cK^{L'}_{(\xi, \zeta)}= \cK^L_{(\xi,\zeta)} + \int_S \iota_\xi\iota_\zeta d\bold{\ell} 
+ \int_S (\iota_\xi\Delta_\zeta \bold{\ell}-\iota_\zeta \Delta_\xi\bold{\ell})\,.
\label{mod_K}
\end{equation}
By means of the very useful identity
 \cite{Freidel:2021cjp}
\be
\d_\zeta \iota_\xi  - \d_\xi \iota_\zeta  + \iota_{\lbr\zeta, \xi\rbr} 
&=d \iota_\zeta  \iota_\xi
-\iota_\zeta  \iota _\xi d
+\iota_\xi\Delta _\zeta 
-  \iota_\zeta \Delta_\xi
+\iota_\xi I_{\d_\zeta}-\iota_\zeta I_{\d_\xi}\,,
\la{rock}
\ee
and a repeated use of  \eqref{ellvarphi}, we can show that
\be
\d \cK^{L'}_{(\xi, \zeta)}&=\d \int_S\left( \iota_\xi\iota_\zeta d\bold{\ell} 
+\iota_\xi\Delta_\zeta \bold{\ell}-\iota_\zeta \Delta_\xi\bold{\ell} \right)
\cr
&=\d \int_S\left(\d_\zeta \iota_\xi \bold{\ell}   - \d_\xi \iota_\zeta  \bold{\ell} + \iota_{\lbr\zeta, \xi\rbr} \bold{\ell} \right)
\cr
&= \int_S\left(\d_\zeta \delta_\xi\bold{\vartheta}  - \d_\xi \delta_\zeta\bold{\vartheta}+ \d_{\lbr\zeta, \xi\rbr} \bold{\vartheta} \right)
\cr
&=0\,,
\ee
where we used  \eqref{maincom} in the last passage.  For the second part, we use the charge transformation \eqref{mod_QF}, together with $
\bold{a}'_\xi = \bold{a}_\xi + \Delta_\xi \bold{\ell}$ and the fact that  the  codimension-2 form $c_{(\xi,\eta)}$ is independent of the choice of boundary Lagrangian, to write
\be
 \{Q'_\xi,Q'_\eta\}_{L'}+\cK^{L'}_{(\xi,\eta)}
  &=  \{Q_\xi,Q_\eta\}_L +\cK^{L}_{(\xi,\eta)}+  \int_S I_{\lbr \xi,\eta\rbr}\bold{\vartheta }
 \cr
&+ 
  \int_S (\delta_\xi  \iota_\eta\bold{\ell }-I_\eta  \delta_\xi  \bold{\vartheta } ) 
 + \int_S \iota_\xi \iota_\eta d \bold{\ell }
+ \int_S \left( \iota_\xi \Delta_{\eta} \bold{\ell } -\iota_{\eta} \Delta_\xi  \bold{\ell } \right)\,.
\ee
We now use \eqref{ellvarphi} and the commutator $[\d_\xi, \iota_{\eta}]=\iota_{\d_\xi \eta}$ to arrive at
 \be
  \{Q'_\xi,Q'_\eta\}_{L'}+\cK^{L'}_{(\xi,\eta)}
 &=  \{Q_\xi,Q_\eta\}_{L} +\cK^{L}_{(\xi,\eta)}+  \int_S I_{\lbr \xi,\eta\rbr}\bold{\vartheta } 
 \cr
&+ \int_S \left( \iota_{\eta} \cL_\xi  \bold{\ell } - \iota_\xi \cL_{\eta} \bold{\ell }
+\iota_{\d_\xi \eta}\bold{\ell }-\iota_{\d_\eta \xi } \bold{\ell }
 +\iota_\xi \iota_\eta d \bold{\ell }
\right)\,.
  \ee
 Application of the identity \eqref{rock}   then yields
 \be
\iota_{\zeta} \cL_\xi  \bold{\ell } -\iota_\xi \cL_{\zeta} \bold{\ell } +
\iota_{\delta_\xi \zeta}\bold{\ell } -\iota_{\delta_\zeta \xi}\bold{\ell } =
-\iota_{\lbr \xi,\zeta\rbr}\bold{\ell } -\iota_\xi \iota_\zeta   d \bold{\ell }
+  d (\iota_\xi \iota_\zeta \bold{\ell })\,,
\la{idL}
\ee
from which
 \be
  \{Q'_\xi,Q'_\eta\}_{L'}+\cK^{L'}_{(\xi,\eta)}=
   \{Q_\xi,Q_\eta\}_{L} +\cK^{L}_{(\xi,\eta)}
   +  \int_S\left( I_{\lbr \xi,\eta\rbr}\bold{\vartheta } -\iota_{\lbr \xi,\eta\rbr}\bold{\ell } \right)\,,
  \ee
and finally 
 \be\label{brack4}
  \{Q'_\xi,Q'_\eta\}_{L'}+\cK^{L'}_{(\xi,\eta)} + Q'_{\lbr \xi,\eta\rbr} = 
  \{Q_\xi,Q_\eta\}_{L} +\cK^{L}_{(\xi,\eta)} + Q_{\lbr \xi,\eta\rbr} \,.
 \ee
\end{proof}

\subsection{Integrable Noetherian flux}\la{sec:F=d}

As the next level of generality, we want to consider the case where the Noetherian flux is given by a total field-variation of some field-space 0-form, say 
\be
\cF_\xi = \delta \cP_\xi\,.
\la{intF}
\ee
 If this condition holds, then the flux contribution can be absorbed in the definition of the charge, \ie
\begin{equation}
-I_\xi\Omega \heq \delta (\cQ_\xi -\cP_\xi). 
\label{ham_ch2}
\end{equation}
We can thus define the {\it Hamiltonian charge}  as
\be
\cH_\xi := \cQ_\xi -\cP_\xi\,,
\label{Hdef}
\ee 
with Poisson bracket given by
\begin{equation}
    \{\cH_\xi, \cH_\zeta\}_L = -I_\xi I_\zeta\Omega = \var_\xi\cH_\zeta\,.
    \label{Poi1}
\end{equation}
In this case where  no non-integrable terms appear in \eqref{ham_ch2},  the previous theorem generalizes as follows:
\begin{theorem}[\textbf{Charge representation theorem 2}]\label{threpr2}
  For Noether charges with integrable Noetherian flux \eqref{intF}, the algebra given by the Poisson bracket in \eqref{Poi1} is isomorphic to the Lie algebra of diffeomorphisms \eqref{mod_lie} up to a central extension, 
    \begin{equation}
     \{\cH_\xi, \cH_\zeta\}_L \heq -\cH_{ \lbr \xi, \zeta \rbr  } - \alpha^L_{(\xi, \zeta)}\,,
     \label{rep2}
    \end{equation}
    where $\alpha$ is a field-independent central extension given by
    \begin{equation}
   \alpha^L_{(\xi,\zeta)}   = \cK^L_{(\xi,\zeta)} -\left(  \var_\zeta\cP_\xi - \var_\xi \cP_\zeta -\cP_{\lbr \xi,\zeta\rbr } \right)\,,
\label{alpha}
\end{equation}
with $\cK^L_{(\xi,\zeta)}$ in \eqref{cocy}, 
    and satisfying  the 2-cocycle identity
    \begin{equation}
        \var_\xi \alpha^L_{(\zeta, \eta)}+ \alpha^L_{(\lbr \xi, \zeta \rbr  , \eta)} + \text{cycl.} =0. \label{id_co-2}
    \end{equation}
\end{theorem}

\begin{proof} 
{
To streamline the proof of the  theorem, it is convenient to start from the identity 
\be
 -I_{\xi}I_{\zeta}\Omega \heq \delta_{\xi}\mathcal{Q}_\zeta-\delta_\zeta\mathcal{Q}_\xi+\mathcal{Q}_{\lbr \xi,\zeta\rbr }+\cK^L_{(\xi,\zeta)}
\ee
derived in \cite{Freidel:2021cjp} (see (3.22) there) and valid for a generic form of the Noetherian flux. For an integrable flux, it then follows that
\begin{equation}
\begin{aligned}
\delta_\xi\mathcal{H}_\zeta&\heq\delta_{\xi}\mathcal{H}_\zeta-\delta_\zeta\mathcal{H}_\xi+\mathcal{H}_{\lbr \xi,\zeta\rbr }+\cK^L_{(\xi,\zeta)}
+\delta_\xi\mathcal{P}_\zeta-\delta_\zeta\mathcal{P}_\xi+\mathcal{P}_{\lbr \xi,\zeta\rbr }.
\end{aligned}
\end{equation}
Hence, one immediately arrives at
\begin{equation}
    \{\mathcal{H_\xi,\mathcal{H}_\zeta}\}_L=\delta_{\xi}\mathcal{H}_\zeta\heq-\mathcal{H}_{\lbr \xi,\zeta\rbr }-\alpha^L_{(\xi,\zeta)},
\end{equation} 
with 
\be
\alpha^L_{(\xi,\zeta)}=
\cK^L_{(\xi,\zeta)}+
\delta_\xi\mathcal{P}_\zeta-\delta_\zeta\mathcal{P}_\xi+\mathcal{P}_{\lbr \xi,\zeta\rbr }.
\la{al}
\ee
To prove that \eqref{al} is a field space constant, we compute
\be
\delta \alpha^L_{(\xi,\zeta)}&=
-\d \cH_{\lbr \xi,\zeta\rbr }
-\d  \var_\xi \cH_\zeta
\cr
 &=I_{\lbr \xi,\zeta\rbr }\Omega
 +\d I_\xi I_\zeta \Omega
 \cr
 &=(I_\zeta\var_\xi -\var_\xi I_\zeta)\Omega
 +\d I_\xi I_\zeta \Omega
 \cr
 &=(I_\zeta \d I_\xi 
- I_\xi \d I_\zeta)\Omega
\cr
 &=0,
\ee
where we used \eqref{dxiI} and $\d^2=0$.
The proof of the 2-cocycle identity follows straightforwardly by substituting $\cQ\to\cH$ in the proof of the previous theorem. 
}
\end{proof}

 From the general relations \eqref{mod_QF}, under a boundary Lagrangian shift that preserves the integrability condition \eqref{intF}, the Noether charge and integrable flux change accordingly to
\begin{equation}
\cQ_\xi' = \cQ_\xi + \int_S (\iota_\xi\bold{\ell} -I_\xi \bold{\vartheta})\,,
\qquad
\cP_\xi' = \cP_\xi +\int_S (\iota_\xi\bold{\ell}-I_\xi\bold{\vartheta} - \bold{\upsilon}_\xi)\,,
\label{mod_QP}
 \end{equation}
 where $\bold{\upsilon}_\xi$ is a field-space 0-form such that
 \be
 I_\xi \d \bold{\vartheta} =\d\bold{\upsilon}_\xi\,.
 \la{ups}
 \ee
The condition above represents the least stringent implementation of the requirement for the   boundary Lagrangian to preserve \eqref{intF}. We thus immediately see that the shifted Hamiltonian charge reads
\be
\cH_\xi' = \cH_\xi+\int_S \bold{\upsilon}_\xi\,.
\la{H'}
\ee
Let us now prove the useful identity
\be
\bold{\upsilon}_{ \lbr \xi, \zeta \rbr  }=-  I_\xi I_\zeta \d \bold{\vartheta}\,.
\la{upxizeta}
\ee
\begin{proof}
{
From  \eqref{mod_QF}, we have
\be
\d\cP'_{ \lbr \xi, \zeta \rbr }  = \d\cP_{ \lbr \xi, \zeta \rbr }  +\int_S (\delta\iota_{ \lbr \xi, \zeta \rbr } \bold{\ell}-\delta_{ \lbr \xi, \zeta \rbr } \bold{\vartheta})\,,
\ee
where
\be
-\delta_{\lbr \xi,\zeta\rbr }\bold{\vartheta}&=[\delta_\xi, \delta_\zeta]\bold{\vartheta}
\cr
&=\delta I_\xi \delta_\zeta \bold{\vartheta}+I_\xi\d^2 \bold{\upsilon}_\zeta -\xi\leftrightarrow\zeta
\cr
&=\delta I_\xi \delta_\zeta \bold{\vartheta} -\xi\leftrightarrow\zeta\,.
\ee
Hence, we can write
\be
\cP'_{\lbr \xi,\zeta\rbr }&=
 \cP_{\lbr \xi,\zeta\rbr }
+\int_S (\iota_{\lbr \xi,\zeta\rbr }\bold{\ell}
+I_\xi\delta_\zeta\bold{\vartheta}-I_\zeta\delta_\xi\bold{\vartheta})\,.
\ee
At the same time, \eqref{mod_QP} yields
\be
\cP'_{\lbr \xi,\zeta\rbr } = \cP_{\lbr \xi,\zeta\rbr } +\int_S (\iota_{\lbr \xi,\zeta\rbr }\bold{\ell}-I_{\lbr \xi,\zeta\rbr }\bold{\vartheta} - \bold{\upsilon}_{\lbr \xi,\zeta\rbr })\,.
\ee
Combining the two expressions for the shifted integrable flux, we have
\be
\bold{\upsilon}_{\lbr \xi,\zeta\rbr }&=I_\zeta\delta_\xi\bold{\vartheta} -I_\xi\delta_\zeta\bold{\vartheta}-I_{\lbr \xi,\zeta\rbr }\bold{\vartheta} 
\cr
&=I_\zeta\delta_\xi\bold{\vartheta} -I_\xi\delta_\zeta\bold{\vartheta}+[\d_\xi,I_\zeta]\bold{\vartheta} 
\cr
&=\d_\xi I_\zeta\bold{\vartheta} -I_\xi\delta_\zeta\bold{\vartheta}
\cr
&=-I_\xi I_\zeta\d\bold{\vartheta}\,,
\ee
where we used the commutator \eqref{dxiI}.
}
\end{proof}

The identity \eqref{upxizeta} can be used to prove in a straightforward manner  the following corollary.

\begin{corollary}
Under a Lagrangian shift that preserves the integrability condition \eqref{intF}, the bracket of the Hamiltonian charge \eqref{Hdef} satisfies
\be
    \{\cH'_\xi, \cH'_\zeta\}_L +\cH'_{ \lbr \xi, \zeta \rbr  }
    &=    \{\cH_\xi, \cH_\zeta\}_L +\cH_{ \lbr \xi, \zeta \rbr  }\,,
\la{Hbrashift}
\ee
and  the central extension $ \alpha_{(\xi,\zeta)}$ is left invariant \footnote{Due to this shift-invariance of the central extension, from now on we drop the superscript $L$ for $\alpha$.} 
\be
\alpha'_{(\xi, \zeta)}= \alpha_{(\xi,\zeta)}\,.
 \label{mod_a}
\ee
\end{corollary}

\begin{proof}
By means of \eqref{H'} and \eqref{upxizeta} we have
\be
    \{\cH'_\xi, \cH'_\zeta\}_L +\cH'_{ \lbr \xi, \zeta \rbr  }&=    \{\cH_\xi, \cH_\zeta\}_L +\cH_{ \lbr \xi, \zeta \rbr  }
    +\int_S\left( \var_\xi  \bold{\upsilon}_\zeta + \bold{\upsilon}_{ \lbr \xi, \zeta \rbr  } \right)
    \cr
    &=    \{\cH_\xi, \cH_\zeta\}_L +\cH_{ \lbr \xi, \zeta \rbr  }
    +\int_S\left(  I_\xi \var \bold{\upsilon}_\zeta 
- I_\xi I_\zeta \d \bold{\vartheta} \right)
    \cr
    &=    \{\cH_\xi, \cH_\zeta\}_L +\cH_{ \lbr \xi, \zeta \rbr  }\,.
\ee
Then \eqref{mod_a} follows from  \eqref{rep2}.
\end{proof}

Let us stress some key features of the algorithm developed above to construct Hamiltonian charges in comparison with other approaches commonly followed in the literature. Once the  pre-symplectic potential is {\it uniquely} derived from a given Lagrangian---via the Anderson's homotopy operator prescription---and the {\it Noetherian split} is implemented, there is no ambiguity left in the determination of the charge and the flux, as long as {\it  all the anomaly contributions} are included in their definitions, as expressed in \eqref{charge_exp}, \eqref{flux}.  This is particularly clear both in the case of a vanishing and an integrable Noetherian flux, as enunciated in the two Charge representation theorems proven above. 

In fact, the content of these two theorems determines a precise algorithm that allows us to compute  the value of the central charge from first principles and with no degree of ambiguity. Such algorithm comprises the following steps:

\begin{enumerate}

\item Once the pre-symplectic potential has been uniquely determined from the given initial Lagrangian, compute the  symplectic form contraction \eqref{ham_ch}.

\item Any time step 1  gives rise to a total variation of a field-space 0-form, this is the quantity that provides a  canonical representation of the vector field algebra. In this case, there are clearly only two options for the Noetherian flux expression: Either it vanishes or it is a  $\d$-exact form. Hence, step 2 consists of computing  the  Noetherian flux  \eqref{flux} to verify which is the case.

\item If the result of  step 2 gives a vanishing flux, the central charge extension is explicitly given by \eqref{cocy}; if the Noetherian flux satisfies \eqref{intF} instead,  the central charge extension is explicitly given by \eqref{alpha}. In either case, one can perform the final step 3, which consists of using the appropriate explicit expression to derive the value of the central charge, {{\it without} the need of a direct calculation of the charge algebra}.

\end{enumerate}

Notice that in the procedure laid down above there is no mention of the Noether charge. Clearly, in step 1, one can instead use  the expression \eqref{charge_exp}, which requires the derivation of Lagrangian anomaly as well,\footnote{Its derivation is in any case required to complete steps 2 and 3.} to derive the Noether charge---this is indeed how we will proceed in Section \ref{alads}. It is then step 2 that informs us whether a Hamiltonian charge exists or not for the theory under considerations. 
Moreover, let us stress once more that the procedure leaves 
no space for any addition of a total variation term in what we call `charge' and what we call `flux', which would follow from  the freedom of adding a $\d$-exact term in the pre-symplectic potential. Such an ambiguity can only be advocated if one omits to include all the anomaly contributions in the definitions \eqref{charge_exp}, \eqref{flux}. Any $\d$-exact term addition to the original pre-symplectic potential can only come from a boundary Lagrangian shift according to
\eqref{theta-shift}, but this is part of the specification of the `given initial Lagrangian', and it can be regarded as step 0, in the spirit of \cite{Freidel:2020xyx,Freidel:2020svx}. In this regard, notice that the property \eqref{mod_a} enunciated in the corollary above, guarantees that, in the case of an integrable flux, we do not need to worry of which representative in the equivalence class $[\bold L]$, where  $\bold L \sim \bold L+d\bold \ell$ {and such that \eqref{ups} is satisfied}, we take to compute the central charge, as this is invariant under any of these boundary Lagrangian shifts; this can turn out to be a very useful property.

This concludes the first part of this work, where we implemented the Noetherian split within the covariant phase space formalism in the presence of field-dependence and anomalies to derive explicit  expressions of the central charge that can arise in the representation of the theory symmetry algebra. We specialized to the cases where the Noetherian flux is either vanishing or integrable. We now apply the procedure described above to a case study where these conditions are met, namely 
to the context  of 3D asymptotically anti-de Sitter spacetimes. We will show how the general form \eqref{alpha} of the central charge exactly reproduces the expression found in previous literature. Before moving to this next step, let us conclude this section with a final remark.

\begin{remark}
When the Noetherian flux \eqref{flux} does not fall in either of the two cases considered in the  Charge representation theorems above, namely when it is a non-vanishing, non-integrable field-space 1-form, one has to obviate to the absence of a Poisson bracket to compute the charge algebra. In this case then, it is now common practice to rely on the
Barnich--Troessaert bracket \cite{Barnich:2010eb, Barnich:2011mi,Barnich:2009se}, defined as
\be
  \{\cQ_\xi,\cQ_\zeta\}_{BT}:= \delta_\xi \cQ_\zeta - I_\zeta \cF_\xi\,.
\label{IBra}
\ee
The definition above can be used for any prescription one uses to derive the charge and the flux, like the Iyer--Wald--Zoupas \cite{Lee:1990nz,Iyer:1994ys,Wald:1999wa} prescription or the  Barnich--Brandt \cite{Barnich:2001jy, Barnich:2003xg, Barnich:2007bf} prescription. However, in none of these cases one can derive an explicit expression for the possible cocycles that can appear in the symmetry algebra representation.

In \cite{Freidel:2021cjp}, it was shown that applying the Barnich--Troessaert bracket to the Noether charge and the Noetherian flux defined as in \eqref{charge_exp}, \eqref{flux} yields a cocycle expression exactly of the form \eqref{cocy}. This result was used to define a new bracket where the term  \eqref{cocy} is included in the expression of the bracket. The advantage of such modification was that the Jacobi identity for the new bracket could easily be proven, while in the case of \eqref{IBra} certain assumptions (see \cite{Barnich:2011mi}) on the brackets involving the a priori field-dependent cocycle \eqref{cocy} are required to guarantee the validity of Jacobi. In fact, the definition \eqref{IBra} only applies to the bracket between two charges, but it doesn't provide a prescription to compute the charge action on a general field-dependent quantity.

However, the content of the Charge representation theorem 1 explicitly shows that the cocycle  \eqref{cocy} appears on the RHS of \eqref{ch_al} {\it already} in the case of a vanishing flux---in which case it becomes a field-independent quantity and the Jacobi identity is automatically satisfied---, where there is no ambiguity in the definition of the charge bracket, as a Poisson bracket is available.  This shows that, in the presence of  a non-vanishing, non-integrable flux, the Noether charge bracket can be computed using the Barnich--Troessaert bracket \eqref{IBra}, with the flux given by  \eqref{flux}. One would then obtain the Charge representation theorem
\be
  \{\cQ_\xi, \cQ_\zeta\}_{BT} \heq -\cQ_{\lbr \xi,\zeta \rbr } - \cK^L_{(\xi, \zeta)}\,,
\ee
with $\cK^L_{(\xi, \zeta)}$  given again by \eqref{cocy}. Validity of the Jacobi identity still rests on the assumptions mentioned above.

{
Let us also point out that a cocycle shift of the form \eqref{alpha} 
was already derived in
\cite{Barnich:2011mi} when shifting the flux by an integrable term (with a compensating shift in the charge). The main difference with our approach is that the {\it reference} cocycle $\cK^L$ is explicitly given by  \eqref{cocy}, while in \cite{Barnich:2011mi} this can be determined only {\it after} the bracket \eqref{IBra} is computed.
}

\end{remark}

\section{Asymptotically Locally AdS$_3$ spacetimes}\label{alads}

We start by briefly reviewing the definition of \emph{asymptotically (locally) anti-de Sitter spacetimes}, focusing on  spacetime dimension equal to three. 
Let us consider a three dimensional manifold $\cM$ endowed with a Lorentzian metric $g_{\mu\nu}$. Following the ambient space  construction introduced by Fefferman and Graham in \cite{Fefferman:2007rka}, an asymptotically anti-de Sitter spacetime can always be put in the following form
\begin{equation}
ds^2= \frac{\ell^2}{r^2} dr^2 + \frac{r^2}{\ell^2} {q}_{ab}(r,x)dx^a dx^b,
\label{ads_metr}
\end{equation}
where $r$ is the space-like radial coordinate and $x^a= (t, \phi)$ are the coordinates defining the two-dimensional cylinder $\Gamma$, with $\phi\sim \phi+2\pi$; $\ell^2  =-1/ \Lambda$ is the AdS$_3$ radius. Therefore, the asymptotic boundary of AdS$_3$ is equivalent to a timelike cylinder $\Gamma$ of coordinates $t, \phi$.

In order to find the asymptotic symmetry group of the 3D anti-de Sitter spacetime, one has to impose appropriate boundary conditions on the asymptotic behaviour of the fields. In the seminal work \cite{Brown:1986nw}, Brown and Henneaux  (BH) took the Minkowski metric as the boundary metric, \ie $q^{(0)}_{ab}=\eta_{ab}$. However, it is well-known that the Penrose conformal compactification implies a conformal class of metrics on the boundary, rather than just one specific metric. Therefore, instead of  picking  a particular representative of this equivalence class, \eg the Minkowski metric $\eta$, 
one can relax this boundary condition, assuming that the boundary metric is conformally flat \cite{Troessaert:2013fma,Alessio:2020ioh}, \ie
\begin{equation}
q_{(0)ab} dx^a dx^b= e^{2\varphi(x)} \eta_{ab}dx^a dx^b = e^{2\varphi(x)}(-dt^2 + \ell^2 d\phi^2)\,,
\label{PBH}
\end{equation}
where $\varphi$ is the so called conformal factor. The boundary condition \eqref{PBH} is sometimes referred to as Penrose--Brown--Henneaux (PBH) boundary condition. Using the Fefferman--Graham ansatz \eqref{ads_metr} (see \cite{ Skenderis:1999nb, Skenderis:2000in,Fiorucci:2020xto}) and imposing Einstein's equations in the decomposition of Riemann tensor yields 
\begin{equation}
\begin{aligned}
&\Bigl(\pa_r +\f{3}{r}\Bigl)\pa_r q_{ab} - \f{1}{2}\pa_r q_{bc}(q^{-1}\pa_r q)^c_{\ a}=0,\\
&\nabla_b \pa_r q_{ac} - \nabla_c\pa_r q_{ab}=0,\\
&\ell^4 R^d{}_{abc} + \Bigl(\f{r^3}{2}(\delta^d_{\ c}\pa_r q_{ab} +(q^{-1}\pa_r q)^d_{\ c} q_{ab}) +\f{r^4}{4}(q^{-1}\pa_r q)^d_{\ c} \pa_r q_{ab}  - (b\leftrightarrow c)\Bigl) =0,
\end{aligned}\label{eeom}
\end{equation}
where $R^d{}_{abc}$ is the Riemann tensor associated with $q_{ab}$.
These equations can be solved order by order in $r$ if we allow the metric to take the form\footnote{One should also consider a logarithmic term in the metric expansion, \ie $h_{[2]ab}r^{-2}\ln r^{-2}$. It has been shown \cite{Skenderis:1999nb, Skenderis:2000in,deHaro:2000vlm} that $h_{[d]}\propto \frac{\delta}{\delta q^{(0)}}\int a_{(d)}$, where $a_{(d)}$ is the conformal anomaly. However, for AdS$_3$ the RHS is vanishing because $\int a_{(2)}\propto \int \sqrt{q^{(0)}}R^{(0)}$ is purely topological.}
\be
q_{ab} = q^{(0)}_{ab}+ \f{\ell}{r}C_{ab}^{(1)}+\f{\ell^2}{r^2} C^{(2)}_{ab}+\cdots.
\label{qexp}
\ee
Thus, by substituting  \eqref{qexp} into \eqref{eeom} one obtains the following exact solution
\begin{equation}
\la{gab}
 {q}_{ab} (r,x)=  q_{ab}^{(0)}(x) + \frac{\ell^2}{r^2}C_{ab}^{(2)}(x)  +\frac{\ell^4}{r^4}C_{ab}^{(4)}\,,
\end{equation}
with
\begin{equation}
C^{(4)}_{ab}(x)\ \hat{=}\ \frac{1}{4} C^{(2)}_{ac}q_{(0)}^{cd}C^{(2)}_{db}, \qquad q^{ab}_{(0)}C_{ab}^{(2)}\ \hat{=}\ -\frac{\ell^2}{2} R^{(0)}, \qquad D^a C_{ab}^{(2)}\ \hat{=}\ -\frac{\ell^2}{2}\pa_b R^{(0)},
\label{ee}
\end{equation}
 and
where we have denoted by $D$ the Levi--Civita connection compatible with the boundary metric $q^{(0)}_{ab}$, and $R^{(0)}$ is the Ricci scalar associated with $q^{(0)}_{ab}$
\be
R^{(0)}= \f2{\ell^2}e^{-2\varphi}(\ell^2 \pa_t^2 - \pa_\phi^2)\varphi
= -2 \Box \varphi
\,.
\la{R0}
\ee
It follows that on-shell
\be
-\ell^2 C_{tt}^{(2)}+C_{\phi\phi}^{(2)} 
\ \hat{=}\ -\f{{\ell^{4}}}{2}e^{2\varphi}R^{(0)}
={\ell^4}e^{2\varphi}\Box \varphi\,.\label{119n}
\ee

Next, we want to find the asymptotic diffeomorphisms which preserve the Fefferman--Graham gauge \eqref{ads_metr} and the boundary conditions \eqref{PBH}. These are given by a vector field $\xi(r, x^a)$ such that
\begin{equation}
\cL_\xi g_{rr} =0, \qquad\cL_\xi g_{ra}=0, \qquad \text{and}\ \qquad \cL_{\xi} g_{ab} = \cO(r^2) \label{fall-off}.
\end{equation}
Writing explicitly the first two conditions, we have
\begin{equation}
 \xi^\alpha \pa_\alpha g_{rr} + 2g_{\alpha r}\pa_r \xi^\alpha=0, \qquad\text{and}\qquad g_{r \alpha}\pa_a \xi^\alpha + g_{\alpha a}\pa_r \xi^\alpha = 0,
\end{equation}
which yield 
\begin{equation}
 \xi^r \pa_r g_{rr} =- 2g_{r r}\pa_r \xi^r \quad \rightarrow \quad \xi^r (r,x)= r \sigma_\xi(x)\\
\end{equation}
and 
\begin{equation}
\begin{aligned}
\pa_r \xi^b= - g^{b a}g_{rr}\pa_a \xi^r \quad 
\rightarrow \quad \xi^b(r,x)= V_\xi^b(x)-\ell^2\pa_a\sigma_\xi(x) \int \frac{dr'}{r'}g^{ab}(r', x) ,
\end{aligned}
\end{equation}
respectively, and where $\sigma_\xi, V_\xi$ are arbitrary field-independent  functions; however, the components $\xi^a$  depend on the  metric field $g^{ab}$ at sub-leading orders. Therefore, the asymptotic vector fields generating the residual diffeomorphisms are
\begin{equation}
\xi^r(r,x) = r \sigma_\xi(x)\,, \qquad     \xi^a (r,x) = V_\xi^a(x) + \f{1}{r^2}\xi^a_{(2)}(x) + \f{1}{r^4}\xi^a_{(4)}(x) + \cO(r^{-6})\,,
\label{54}
\end{equation}
where
\begin{equation}
\begin{aligned}
    \xi^a_{(2)} =\f{\ell^4}{2}q^{ab}_{(0)}\pa_b \sigma_\xi, \qquad
    \xi^a_{(4)} = -\f{\ell^6}{4}q^{ca}_{(0)}C^{(2)}_{cd}q^{bd}_{(0)}\pa_b\sigma_\xi\,.
 \la{xia}
\end{aligned}
\end{equation}
We come back to the implications  of the third fall-off condition in \eqref{fall-off} in a moment.

The vectors   \eqref{54}  form a closed algebra under the bracket  \eqref{mod_lie} given by\footnote{\la{foot} Not to clutter the notation, we denote $V_1=V_{\xi_1},
\sigma_1=\sigma_{\xi_1}$.} \cite{Barnich:2010eb, Compere:2020lrt} (a proof up to $\cO(r^{-2})$ is also given in Appendix \ref{app:vec-alg})
\be
\lbr \xi_1(V_1, \sigma_1), \xi_2(V_2, \sigma_2) \rbr  = \xi(V_{12}, \sigma_{12})\,,
\label{alg_symm}
\ee
with
\begin{equation}
    {V}^t_{12}=V^a_1\pa_a V^t_2 -V^a_2\pa_a V^t_1, \qquad 
{\sigma}_{12}=V^a_1\pa_a \sigma_2 -V^a_2\pa_a \sigma_1. 
\la{V12}
\end{equation}

\subsection{Variations of the solution space}
We can study how the metric components transform under the action of the vector fields $\xi$ defined above in \eqref{54}, \eqref{xia}.
 The general transformation rule of the metric functionals $\Phi$ reads as follows
\begin{equation}
    \var_\xi \Phi = \cL_\xi \Phi + \Delta_\xi\Phi,
\end{equation}
where the last term represents the anomaly contribution. Since the spacetime metric transforms covariantly under the symmetry transformations, we have that
\begin{equation}
\delta_\xi g_{ab} = \frac{r^2}{\ell^2} \Bigl(\delta_\xi q^{(0)}_{ab} + \frac{\ell^2}{r^2} \delta_\xi C_{ab}^{(2)} + \frac{\ell^4}{r^4} \delta_\xi C_{ab}^{(4)}\Bigl)
\end{equation}
must be equal to $\cL_\xi g_{\mu\nu}$. 
Putting together the contributions of the same order implements the  third fall-off condition in \eqref{fall-off}, and  we  obtain the following transformation rules
\begin{subequations}
\begin{align}
\begin{split}
\delta_\xi q^{(0)}_{ab} &=  \cL_{V_\xi} q_{ab}^{(0)} + 2\sigma_\xi q^{(0)}_{ab},\label{69a}
\end{split}\\
\begin{split}
\delta_\xi C^{(2)}_{ab}
&=\cL_{V_\xi} C_{ab}^{(2)} +\ell^2\left(q^{(0)}_{ab}q_{(0)}^{cd}\pa_c\sigma_\xi \pa_d \varphi +  \pa_b\pa_a\sigma_\xi -2\pa_{(a}\varphi \pa_{b)}\sigma_\xi\right)
\label{66b},
\end{split}\\
\begin{split}
\delta_\xi C^{(4)}_{ab} &=  \cL_{V_\xi} C_{ab}^{(4)} -2\sigma_\xi C_{ab}^{(4)}+
\frac{\ell^2}{2}q^{cd}_{(0)}\pa_d\sigma_\xi \pa_c C^{(2)}_{ab} + \frac{\ell^2}{2}C_{ac}^{(2)}\pa_b(q^{cd}_{(0)}\pa_d\sigma_\xi)\\
&\quad + \frac{\ell^2}{2}C_{bc}^{(2)}\pa_a(q^{cd}_{(0)}\pa_d\sigma_\xi) - \frac{\ell^2}{4}q^{ce}_{(0)}C_{ef}^{(2)}q^{df}_{(0)}\pa_d\sigma_\xi \pa_c q^{(0)}_{ab} \\
&\quad- \frac{\ell^2}{4}q_{ca}^{(0)}\pa_b(q^{ce}_{(0)}C_{ef}^{(2)}q^{df}_{(0)}\pa_d\sigma_\xi) - \frac{\ell^2}{4}q_{cb}^{(0)}\pa_a(q^{ce}_{(0)}C_{ef}^{(2)}q^{df}_{(0)}\pa_d\sigma_\xi).
\end{split}
\end{align}
\label{var_func}
\end{subequations}

The main difference between BH and PBH boundary conditions is that, in the latter case, we allow for variations of $q^{(0)}_{ab}$ through variations of the conformal factor $\varphi$, \ie $\delta q^{(0)}_{ab}=2(\delta\varphi)q^{(0)}_{ab}$. Therefore,  equation \eqref{69a} becomes
\begin{equation}
D_a V_{\xi b} + D_b V_{\xi a} = -2(\omega_\xi +\sigma_\xi )q^{(0)}_{ab} \,,
 \label{Kill_eq1}
\end{equation}
where we have defined  the Weyl  parameter $\omega_\xi(x) = -\delta_\xi \varphi(x)$.  
Introducing 
 \be
 \chi_{V_\xi}(x) := \sigma_\xi(x)+\omega_\xi(x)\,,
 \la{chi}
 \ee
and  tracing  \eqref{Kill_eq1}, we have 
\begin{equation}
\chi_{V_\xi} = -\frac{1}{2} D_a V_\xi^a
= -\frac{1}{2} \pa_a V_\xi^a-V^a_\xi\pa_a \varphi
\,.
\label{56}
\end{equation}
It follows that $V_\xi^a$ are conformal Killing vectors of the boundary metric $q^{(0)}_{ab}$ and \eqref{Kill_eq1} implies the conditions 
\be
\pa_t V_\xi^t=\pa_\phi V_\xi^\phi\,,\quad \pa_\phi V_\xi^t=\ell^2 \pa_t V_\xi^\phi\,,
\la{VKill}
\ee
namely, upon introduction of light-cone coordinates $x^{\pm} = (t/\ell)\pm \phi$, $V_\xi^\pm$  are chiral vectors.

When restricting to ordinary Brown--Henneaux boundary conditions, the LHS of \eqref{69a} is zero and therefore
\begin{equation}
    \sigma_\xi = -\f{1}{2}\pa_a V_\xi^a.
\end{equation}

So far, the Weyl parameter $\omega_\xi$ has been considered as a generic function of $x^a$. However, if we introduce 
the holographic stress-energy tensor  \cite{Henningson:1998gx,Balasubramanian:1999re}
\begin{equation}
T_{ab} 
= \frac{1}{\ell} \Bigl(C^{(2)}_{ab} +\frac{\ell^2}{2}q^{(0)}_{ab}R^{(0)}\Bigl)\,,
\label{T}
\end{equation}
and require that its trace transforms homogeneously, the harmonicity condition
\be
\Box \omega_\xi=0\,
\label{boxomega}
\ee
follows \cite{Alessio:2020ioh}.
Hence we can write the function $\omega_\xi$ in light-cone coordinates as 
\be
\omega_\xi(x^+, x^-)=\omega_\xi^+ (x^+) + \omega_\xi^- (x^-)\,.
\la{omega+-}
\ee

\subsection{Lagrangian formulation}\label{EH_sec}

In order to implement the covariant phase space formalism described in Section \ref{sec:ECS}, a lagrangian formulation needs to be selected.
In this work, we will use the metric formulation of General Relativity. The Einstein--Hilbert (EH) action for three-dimensional gravity with negative cosmological constant is
\begin{equation}
S[g] = \frac{1}{2} \int_{\cM} (R[g]-2\Lambda)\bold{\epsilon},\label{EHaction}
\end{equation} 
where   $\bold{\epsilon}$ is the volume form, which explicitly reads
\begin{equation}
    \bold{\epsilon} = re^{2\varphi} \Bigl[1+ \frac{1}{2r^2}e^{-2\varphi} \Bigl(C_{\phi\phi}^{(2)} -\ell^2 C_{tt}^{(2)}  \Bigl) + \cO(r^{-4})\Bigl] dt\wedge dr\wedge d\phi\,.
\end{equation}
The on-shell Einstein--Hilbert Lagrangian is
\begin{equation}
    L _{\EH}= \f{1}{2}(R-2\Lambda)\sqrt{-g} \heq -\f{2}{\ell^2}re^{2\varphi}\Bigl(1 -\f{\ell^4}{4r^2}R^{(0)}\Bigl) +\cO(r^{-3})\,,
    \la{LEH}
\end{equation}
and the 3D Ricci scalar is
\begin{equation}
R[g] = -\f{6}{\ell^2} + \f{2}{\ell^2 r^2} e^{-2\varphi} \left(\frac{\ell^4}{2}e^{2\varphi}R^{(0)} + C_{\phi\phi}^{(2)} - \ell^2 C_{tt}^{(2)} \right) +\cO(r^{-4})\,.
\end{equation}
The variation of the Lagrangian yields
\begin{equation}
\delta \bold{L}_{\EH}[g, \delta g] = \frac{1}{2} \bold{\epsilon} \Bigl(R_{\mu\nu} - \frac{1}{2} Rg_{\mu\nu}+ \Lambda g_{\mu\nu}\Bigl)\delta g^{\mu\nu} + \bold{\epsilon}\nabla_\mu \theta^\mu_{\EH}[g, \delta g]\,,
\end{equation}
where $\theta^\mu_{\EH}$ is the pre-symplectic potential  current of  EH gravity given by
\be\label{EH current}
\theta^\mu_\EH=\f{1}{2}(g^{\mu\nu}\nabla^\sigma \var g_{\nu\sigma} - g^{\nu\sigma}\nabla^\mu \var g_{\nu\sigma})\,.
\ee
In order to have a finite action and a well-defined variational principle, the action \eqref{EHaction} has to be supplemented with the Gibbons--Hawking--York boundary term \cite{Brown:1992br, Gibbons:1976ue} plus a counter-term $K_0=2/\ell$, \ie
\begin{equation}
    S[g]=\f{1}{2}\int_\cM (R[g]-2\Lambda)\eps -\f{1}{2}\int_\Gamma (2K-K_0)\beps\,,
\label{GHact}
\end{equation}
where $\beps$ is the  volume form on the time-like boundary  $\Gamma$ and $K$ is the trace of its extrinsic curvature, given by 
\begin{equation}
{K}=\frac{2}{\ell}+\frac{r}{2\ell}q^{ab}\partial_rq_{ab}= \frac{2}{\ell} + \frac{1}{\ell r^2}e^{-2\varphi} (\ell^2 C^{(2)}_{tt}- C_{\phi\phi}^{(2)}) +\cO(r^{-4})\,.
\label{KPBH}
\end{equation}
When  \eqref{PBH} holds, the variational principle is well-defined if and only if the following condition \cite{Troessaert:2013fma}
\begin{equation}
    K-K_0 = o(r^{-2})
    \label{nbc}
\end{equation}
is imposed. Indeed, varying the regularized action \eqref{GHact} yields
\begin{equation}
    \var S[g] = \f{1}{2}\int_\cM (G_{\mu\nu} + \Lambda g_{\mu\nu})\var g^{\mu\nu}\eps - \int_\Gamma \frac{r^2}{\ell^2}\var\varphi (K-K_0) \sqrt{-q^{(0)}}\ dt\wedge d\phi\,,
\end{equation}
and to have a well-posed variational principle the leading order of the boundary term must be zero. This implies that $g^{\mu\nu}\var K_{\mu\nu} + K_{\mu\nu}\var g^{\mu\nu} = o(r^{-2})$, which leads to mixed boundary conditions. Thus, from \eqref{KPBH} and \eqref{119n}, the new boundary condition \eqref{nbc} implies
\begin{equation}
    q^{ab}_{(0)} C_{ab}^{(2)}=0 \Rightarrow \Box \varphi=0\,;
    \label{boxvarphi=0}
\end{equation}
furthermore, variation of the constraint \eqref{nbc}, $\var K= o(r^{-2})$, yields
\begin{equation}
     \var\Box\varphi = -(2\var\varphi)\Box\varphi + \Box \var\varphi=0 \Rightarrow \Box\var\varphi=0\,.
    \label{boxvarphi}
\end{equation}
This condition is the same as \eqref{boxomega}. Namely, requiring the covariance for $T^a{}_{a}$ is equivalent to imposing the falloff condition \eqref{nbc} to the extrinsic curvature. 

As already pointed out in \cite{Troessaert:2013fma}, the boundary condition \eqref{boxvarphi=0} is essential in order for the symplectic flux to be integrable; we confirm this 
in Appendix \ref{CF}.

Moreover, as shown in Appendix \ref{CQ}, the condition \eqref{boxvarphi} together with \eqref{chi} and \eqref{VKill} imply
\be
\Box\sigma_\xi=0\,.
\la{boxs}
\ee
It follows that $\sigma_\xi$ is also a harmonic function and the algebra in \eqref{alg_symm} consists of the  semi-direct sum of the two-dimensional chiral conformal algebra with harmonic Weyl transformations.

Finally, using the third equation in \eqref{ee}, the supplementary boundary condition \eqref{nbc} also implies that
\begin{equation}
D^a C_{ab}^{(2)}=0\,.
\end{equation}

In order to proceed with the Noether analysis, we give the expressions of some necessary geometric quantities. Using the rules discussed in Section \ref{NC}, the contractions of the volume form read as follows
\be
&\teps= \left(
\ell e^{\varphi} 
+\f{\ell}{2r^2}e^{-\varphi}C_{\phi\phi}^{(2)}
\right)
\, dr\wedge d\phi+\cO(r^{-4}),
\cr
&\beps=  \frac{r^2}{\ell}e^{2\varphi}\Bigl(1+\frac{1}{2r^2} e^{-2\varphi}(C_{\phi\phi}^{(2)}-\ell^2 C_{tt}^{(2)})+\cO(r^{-4})\Bigl) dt\wedge d\phi,
\cr
&\bar{\teps}= -\left(
r e^{\varphi} 
+\f{1}{2r}e^{-\varphi}C_{\phi\phi}^{(2)}
\right)
\,  d\phi+\cO(r^{-3}),
\ee
where in particular 
  \begin{equation}
  \begin{aligned}
  n^\mu \pa_\mu &= -\Bigl( \frac{\ell}{r} e^{-\varphi}  + \frac{\ell^3}{ r^3}e^{-3\varphi}   C_{tt}^{(2)} +\cO(r^{-5})\Bigl) \pa_t  + \Bigl( \frac{\ell}{2 r^3} e^{-3\varphi} C_{t\phi}^{(2)}+\cO(r^{-5})\Bigl)\pa_\phi,\\
   n_\mu dx^\mu &= \frac{r}{\ell}e^{\varphi}\Bigl(1 -\frac{\ell^2}{2r^2}e^{-2\varphi} C_{tt}^{(2)}  +\cO(r^{-4}) \Bigl) dt,
    \end{aligned}
  \end{equation}
are respectively the normal vector to $\Sigma$ and the relative 1-form, and
\begin{equation}
        s^\mu \pa_\mu  = \frac{r}{\ell}\pa_r,\qquad   s_\mu dx^\mu = \frac{\ell}{r}dr,
\end{equation}
are respectively the normal vector to $\Gamma$ and its relative 1-form.

\subsection{Charges and fluxes}

As we will see in a moment, asymptotically locally AdS$_3$ spacetimes with PBH boundary conditions correspond to systems with an integrable Noetherian flux. The boundary Lagrangian shift introduced in \eqref{GHact} gives rise to a vanishing corner potential, {which means that \eqref{ups} is trivially satisfied}. Therefore, in light of  Section \ref{sec:F=d}, we can safely compute the Hamiltonian charge and the associated algebra central extension by means of the EH Lagrangian formulation, as the final expressions will be the same than with the shifted Lagrangian.

 Using the charge expression in \eqref{charge_exp}, the Noether charge associated with the vector fields  \eqref{54} is computed deriving first the Lagrangian anomaly (see Appendix \ref{CQ})
 \be
\bold{a}_\xi
&=
\Bigl(\f{2r^2}{\ell^2}  e^{2\varphi} \var_\xi\varphi {+\cO(r^{-2})}\Bigl) dt\wedge d\phi
+\f2{\ell^2} \xi^\mu \bold{\epsilon}_\mu
\,,
\label{Lanomaly}
\ee
 from which
 \be
\iota_\xi \bold{L} +\bold{a}_\xi&=
-\f{2r^2}{\ell^2}  e^{2\varphi} \omega_\xi dt\wedge d\phi\,.
\la{ixiLaxi}
\ee
The first contribution on the RHS of  \eqref{charge_exp} is given by
 \be
 I_\xi\bold{\theta}
 &=- d\left( \f12  \bar \teps s_\sigma n_\mu  (\nabla^{\sigma} \xi^{\mu} - \nabla^{\sigma} \xi^{\mu}) \right)  - \frac{2}{\ell^2}   \eps_\mu \xi^\mu\,.
 \la{Ixitheta}
\ee
combining the last term on the RHS above with \eqref{ixiLaxi} gives
\be
-\iota_\xi \bold{L} - \bold{a}_\xi - \frac{2}{\ell^2}\xi^\mu \eps_\mu
&=
- \frac{1}{\ell^2}d \left(r^2  e^{2\varphi} V_\xi^t  d\phi\right)
+ \frac{1}{\ell^2}d \left(r^2  e^{2\varphi} V_\xi^\phi   dt\right)
\cr
&+{\ell^2}d \left(\ln{r}\pa_t\sigma_\xi  d\phi\right) +d \left(\ln{r} \pa_\phi \sigma_\xi dt \right)\,,
\la{ixiL+}
\ee
 When integrating over the corner, the  partial-$t$ derivative of $\sigma_\xi$ term of \eqref{ixiL+} vanishes,\footnote{This can easily be seen using Fourier expansion of  the general solution of  $\Box\sigma_\xi=0$.} the divergent term gets cancelled and we arrive at the final expression
(see Appendix \ref{CQ})
\be
\cQ_\xi  &=\int_0^{2\pi} d\phi \   V^a_\xi C^{(2)}_{ta}\,.
\ee

Let us now compute the flux content of the theory. The leakage of the symplectic information through the time-like boundary $\Gamma$ is encoded in the Noetherian flux 
\be
    \cF_\xi &= \int_\cS (\iota_\xi \bold{\theta} +\bold{q}_{\var\xi} + \bold{A}_\xi)\,.
    \la{Fxi}
\ee
As shown in Appendix \ref{CF}, one first computes the pre-symplectic potential anomaly and finds
\be
\Delta_\xi  \bold{\theta} =0\,,
\la{Dtheta}
\ee
as expected.
Therefore, we have
\be
d(\bold{q}_{\var\xi} + \bold{A}_\xi)&=
I_{\var \xi}\bold{\theta} - \iota_{\var \xi}\bold{L} 
     -\var \bold{a}_\xi \,,
     \la{dqA}
\ee
and this evaluates to (see Appendix \ref{CF})
\be
d(\bold{q}_{\var\xi} + \bold{A}_\xi)
&=-\d d\left( \f{r^2}{\ell^2} e^{2\varphi} V^t_\xi d\phi\right)
+\d d\left( \f{r^2}{\ell^2} e^{2\varphi} V^\phi_\xi dt\right)
\cr
&+\d d \left( \ell^2 \ln{r} \pa_t \sigma_\xi  d\phi+\ln{r}\pa_\phi \sigma_\xi dt\right)\,.
\ee
When integrating over the corner, the  partial-$t$ derivative of $\sigma_\xi$ term vanishes as above, and we get
\be
    \cF_\xi &= \int_\cS \iota_\xi \bold{\theta} 
    -\d \left( \int_\cS   \f{r^2}{\ell^2} e^{2\varphi} V^t_\xi d\phi \right)\,.
    \la{Fxi-3}
\ee
Finally, one computes the symplectic flux contribution in \eqref{Fxi-3} and, as shown in Appendix \ref{CF}, the boundary condition \eqref{boxvarphi=0} allows us to remove a non-integrable term. The final expression of the symplectic flux cancels the second divergent term in \eqref{Fxi-3} and yields
\be
    \cF_\xi &= \delta\left(
\ell^2 \int_0^{2\pi}d\phi\ 
( \sigma_\xi \partial_t \varphi
-\varphi \partial_t \sigma_\xi)
 \right)\,.
 \la{Fxi-4}
    \ee

We thus obtain an  integrable Noetherian flux, namely $ \cF_\xi =\d  \cP_\xi $ with
\be
  \cP_\xi = \ell^2  \int_0^{2\pi}d\phi\ 
( \sigma_\xi \partial_t \varphi
-\varphi \partial_t \sigma_\xi)\,,
\la{Pxi}
\ee
and the Hamiltonian charge is given by
\be
\cH_\xi&=\int_0^{2\pi} d\phi \   V^a_\xi C^{(2)}_{ta}
-\ell^2 \int_0^{2\pi}d\phi\ 
( \sigma_\xi \partial_t \varphi
-\varphi \partial_t \sigma_\xi)\,.
\label{H}
\ee

\subsection{Charge algebra and central charge}

Having derived the Hamiltonian charge corresponding to the PBH boundary conditions, we want to verify that the Charge representation theorem 2 holds, namely that
 its (field space) Lie derivative satisfies \eqref{rep2}, where the central charge $\alpha$ is given by \eqref{alpha}.

Using \eqref{66b},
we derive first\footnote{We stick to the convention of footnote \ref{foot}.} (see Appendix \ref{Qalgebra-app})
\be
\var_{\xi_2} \cQ_{\xi_1}
&=\ell^2\int_0^{2\pi} d\phi \  \bigg[-\f1{\ell^2} \cL_2 V^a_1 C^{(2)}_{ta} 
-V^a_1\pa_{a}\sigma_2\pa_{t}\varphi 
 +\varphi
\pa_{t} \left(
V^a_1\pa_a\sigma_2
\right)
 +  V^a_1\pa_t\pa_a\sigma_2 
\bigg]\,.
\la{dxiQ}
\ee
Next, after a few manipulations reported in Appendix \ref{Qalgebra-app}, we have
 \be
 -\var_{\xi_2} \cP_{\xi_1}&= 
 \ell^2\int_0^{2\pi}d\phi\ \left[
V^a_2\pa_a \sigma_1 \partial_t \varphi
-\varphi
\pa_{t} \left(
V^a_2\pa_a\sigma_1
\right)\right]
\cr
&-
 \ell^2 \int_0^{2\pi}d\phi\ \left[ \sigma_1 \partial_t \sigma_2 -\sigma_2 \partial_t \sigma_1
+V^a_2\pa_t\pa_a\sigma_1\right]\,.
\la{dxiP}
 \ee
Putting \eqref{dxiQ} and \eqref{dxiP} together, we arrive at
\be
\var_{\xi_2} \cH_{\xi_1}&= 
\int_0^{2\pi} d\phi \  \bigg[\cL_1 V^a_2 C^{(2)}_{ta} 
+
\ell^2\left((V^a_2\pa_a \sigma_1-V^a_1\pa_{a}\sigma_2)\pa_{t}\varphi 
 -\varphi
\pa_{t} \left(
V^a_2\pa_a\sigma_1
-
V^a_1\pa_a\sigma_2
\right)\right)\bigg]
\cr
& 
-\ell^2 \int_0^{2\pi}d\phi\ \bigg[ \sigma_1 \partial_t \sigma_2 
- V^a_1\pa_t\pa_a\sigma_2 
-1\leftrightarrow2
\bigg]
\cr
& =-  \cH_{\lbr \xi_2,\xi_1\rbr}
-\ell^2 \int_0^{2\pi}d\phi\ \bigg[ \sigma_1 \partial_t \sigma_2- V^a_1\pa_t\pa_a\sigma_2  -1\leftrightarrow2\bigg]\,,
 \ee
with the vector field bracket  given in \eqref{alg_symm}, \eqref{V12}.

Therefore, the Hamiltonian charge \eqref{H} provides a canonical representation of the asymptotic symmetry algebra of AdS$_3$ spacetimes with PBH boundary conditions, with a  central charge given by
\be
\alpha_{(\xi_2, \xi_1)}= \ell^2 \int_0^{2\pi}d\phi\ \bigg[ \sigma_1 \partial_t \sigma_2- V^a_1\pa_t\pa_a\sigma_2  -1\leftrightarrow2\bigg]\,,
 \la{ce}
\ee
and reproducing the expression found in \cite{Troessaert:2013fma}.

In order to verify the validity of the  Charge representation theorem 2  \eqref{rep2}, we are left to show that the  expression \eqref{alpha} for $\alpha_{(\xi_2, \xi_1)}$ reproduces \eqref{ce}.
To achieve this, let us point out that \eqref{dxiP}  implies
 \be
 \var_{\xi_2}\cP_{\xi_1} - \var_{\xi_1} \cP_{\xi_2}
 &=-\ell^2\int_0^{2\pi}d\phi\ \left(
V^a_2\pa_a \sigma_1 \partial_t \varphi
-\varphi
\pa_{t} \left(
V^a_2\pa_a\sigma_1
\right)-1\leftrightarrow2\right)
 \cr
 &
+\ell^2 \int_0^{2\pi}d\phi\ \left( 2 \sigma_1 \partial_t \sigma_2- V^a_1\pa_t\pa_a\sigma_2  -1\leftrightarrow2\right)
 \cr
 &=-\cP_{\lbr \xi_2,\xi_1 \rbr}+\alpha_{(\xi_2, \xi_1)}
 +\ell^2 \int_0^{2\pi}d\phi\ \left( \sigma_1 \partial_t \sigma_2  -1\leftrightarrow2\right)\,.
 \la{dP12}
 \ee
Therefore, consistency with the general expression \eqref{alpha} requires us to show that the cocycle $\cK^L_{(\xi_2, \xi_1)}$ in \eqref{cocy} takes the form
\be
\cK^L_{(\xi_2, \xi_1)}=\ell^2 \int_0^{2\pi}d\phi\ \left( \sigma_2 \partial_t \sigma_1  -1\leftrightarrow2\right)\,.
\la{K12}
\ee
This is done in Appendix \ref{L_an}, where we first derive
\be
 d\bold{c}_{(\xi_2,\xi_1)}& = 
-d\iota_{\xi_2} \bold{a}_{\xi_1}
  + \frac{2}{\ell^2}d \left( 
  \sqrt{-g}   \xi_2^r V_1^a \varepsilon_{ab}\  dx^b
\right)
+\frac{2}{\ell^2}d\left(\sqrt{-g}\xi_2^t\xi_1^\phi\ dr \right) -1\leftrightarrow 2\,,
\la{c}
 \ee
  where $\varepsilon_{t\phi}=1$. The first term on the RHS of \eqref{c} cancels the second corner integral on the RHS of  \eqref{cocy} and, combined with
$  \int_{\cS}\iota_{\xi_2}\iota_{\xi_1} \bold{L}_{\EH}$, the corner integral of $\bold{c}_{(\xi_2,\xi_1)}$ yields exactly the final expression \eqref{K12}.

To summarize, the expression \eqref{K12} for the cocycle, together with 
the relation \eqref{dP12}, reproduces  the expression of the  central charge in the case of an integrable flux, namely
\be
\alpha_{(\xi_2, \xi_1)}&= \cK^L_{(\xi_2, \xi_1)}
 +\left(  \var_{\xi_2}\cP_{\xi_1} - \var_{\xi_1} \cP_{\xi_2} +\cP_{\lbr \xi_2,\xi_1 \rbr}\right)\,.
\ee

 In the case of BH boundary conditions, the Charge representation theorem 1  \eqref{ch_al} is  recovered instead, as in this case the Noetherian flux \eqref{Fxi-4} is vanishing. 
More precisely, for $\varphi=0, \sigma_\xi=-\tfrac 12 \pa_a V^a_\xi$, the centrally extended algebra reads
 \be
 \var_{\xi_2} \cQ_{\xi_1}=-  \cQ_{\lbr \xi_2,\xi_1\rbr}-\cK^L_{(\xi_2, \xi_1)}\,,
 \ee
 with
 \be
 \cK^L_{(\xi_2, \xi_1)}=\alpha_{(\xi_2, \xi_1)}=
 \ell^2 \int_0^{2\pi}d\phi\ \bigg[ \sigma_2 \partial_t \sigma_1 -1\leftrightarrow2\bigg]\,,
 \ee
 where we used 
 \be
 \int_0^{2\pi}d\phi\ \bigg[ - V^a_1\pa_t\pa_a\sigma_2  -1\leftrightarrow2\bigg]
 =\int_0^{2\pi}d\phi\ \bigg[ -2 \sigma_1 \pa_t\sigma_2  -1\leftrightarrow2\bigg]\,,
 \ee
 valid for BH boundary conditions.

\section{Conclusion}

In this work we have investigated the canonical realization of diffeomorphism symmetries within the ambiguity-free formulation of the covariant phase space in the presence of field-dependent vector fields and Lagrangian anomalies. Building on the Noetherian split introduced in \cite{Freidel:2020xyx,Freidel:2021cjp},  once a Lagrangian is specified, the associated pre-symplectic potential, Noether charge and Noetherian flux are uniquely determined. This removes the traditional ambiguities associated with exact shifts of the pre-symplectic potential and provides a canonical geometric framework for the construction of symmetry generators.

Our main result consists in proving two Charge Representation Theorems describing the canonical realization of the diffeomorphism algebra in the two physically relevant situations in which the Noetherian flux either vanishes or is integrable. In the first case, the Noether charge itself generates the symmetry algebra through the canonical Poisson bracket and acquires a field-independent central extension. In the second case, the Hamiltonian generator is obtained by subtracting the integrable Noetherian flux from the Noether charge, leading again to a canonical representation of the symmetry algebra together with an explicit expression for the corresponding central extension. 

Beyond the technical derivation of these results, we believe that their main significance is conceptual. Traditionally, central extensions are obtained through an explicit computation of the charge algebra once an appropriate notion of charge has been selected. Within the ambiguity-free covariant phase space, the logic is instead reversed. The fundamental ingredients are the Lagrangian and the unique pre-symplectic potential associated with it. These determine unambiguously the Noetherian split and, consequently, the appropriate notion of Hamiltonian generator. {The central extension is then no longer an additional quantity that must be fixed by comparison with the charge algebra. Rather, it emerges as an intrinsic quantity, entirely determined from Lagrangian formulation of the theory by the geometric structures defining  the ambiguity-free phase space.}

From this perspective, the two Charge Representation Theorems establish a universal algorithm for the determination of central charges. Once the Noetherian split has been identified, the only information required is whether the corresponding Noetherian flux vanishes or is integrable. In the former case, the central extension is directly given by the cocycle \eqref{id_co} defined solely by the Lagrangian and its associated anomaly. In the latter, it is uniquely modified by the integrable flux contribution \eqref{alpha}. No explicit computation of the charge algebra is required.

As a non-trivial application of this general construction, we considered three-dimensional Einstein gravity with negative cosmological constant and Penrose--Brown--Henneaux boundary conditions. This example provides a particularly rich testing ground because field-dependent asymptotic symmetries, Lagrangian anomalies and an integrable Noetherian flux coexist. Applying the general formalism developed in the first part of this work, we derived the Hamiltonian generators directly from the Einstein--Hilbert Lagrangian and recovered the known central extension of the asymptotic symmetry algebra from the general expression obtained before. Ordinary Brown--Henneaux boundary conditions then appear naturally as the limiting case in which the Noetherian flux vanishes, providing an explicit realization of the first Charge Representation Theorem.

Although the present work focuses on asymptotically locally AdS$_3$ spacetimes, the construction is completely general and applies to any diffeomorphism-invariant theory admitting field-dependent symmetry generators and Lagrangian anomalies. We therefore expect the framework developed here to provide a useful starting point for the study of asymptotic symmetry algebras in a much broader class of theories, including higher-dimensional asymptotically AdS spacetimes, asymptotically flat gravity coupled to matter and gauge fields, theories with gravitational Chern--Simons terms, and more generally any diffeomorphism-invariant theory in which a (on-shell) non-vanishing Lagrangian and an associated anomaly appear.

\section*{Acknowledgements}

We would like to thank Stefano Ansoldi, Gianfranco De Simone and Domenico Salluce for discussions on this subject.

\appendix

\section{Vector field bracket}\label{app:vec-alg}

The Lie bracket components of the vector fields \eqref{54} are given by
\be
[\xi_1, \xi_2]^r_{\text{Lie}}&= \xi_1^\mu \pa_\mu \xi_2^r- 1\leftrightarrow 2
=r V_1^a \pa_a \sigma_2- 1\leftrightarrow 2\,,
\\
[\xi_1, \xi_2]^a_{\text{Lie}}&= \xi_1^\mu \pa_\mu \xi_2^a- 1\leftrightarrow 2
\cr
&=-  \frac{\ell^4}{r^2} e^{-2\varphi} \eta^{ab}\sigma_1\pa_b \sigma_2  
+V_1^b \pa_b V_2^a
\cr
&+ \frac{\ell^4}{2r^2}e^{-2\varphi}  \eta^{ac}V_1^b \pa_b \pa_c \sigma_2
- \frac{\ell^4}{r^2}e^{-2\varphi}  \eta^{ac}V_1^b \pa_b \varphi \pa_c \sigma_2
+ \frac{\ell^4}{2r^2}e^{-2\varphi} \eta^{bc}\pa_c \sigma_1  \pa_bV_2^a
- 1\leftrightarrow 2\,,
\cr
\ee
and the field variation components are given by
\be
\d_2 \xi_1^r&=0\,,
\\
\d_2 \xi_1^a
&=\frac{\ell^4}{ r^2}e^{-2\varphi}\eta^{ab} \omega_2 \pa_b \sigma_1\,.
\ee
We can thus write the components of vector field bracket \eqref{mod_lie} 
\be
\lbr \xi_1, \xi_2 \rbr^r &:=r V_1^a \pa_a \sigma_2- 1\leftrightarrow 2\,,
\\
 \lbr \xi_1, \xi_2 \rbr^a &:=
 V_1^b \pa_b V_2^a
-  \frac{\ell^4}{r^2} e^{-2\varphi} \eta^{ab}(\sigma_1+\omega_1)\pa_b \sigma_2  
\cr
&+ \frac{\ell^4}{2r^2}e^{-2\varphi}\left(  \eta^{ac}V_1^b \pa_b \pa_c \sigma_2
-2\eta^{ac}V_1^b \pa_b \varphi \pa_c \sigma_2
+\eta^{bc}\pa_c \sigma_1  \pa_bV_2^a\right)
- 1\leftrightarrow 2
\cr
&= V_1^b \pa_b V_2^a
+\frac{\ell^4}{2r^2}e^{-2\varphi}\left( 
\eta^{ab} \pa_b (V_1^c \pa_c \sigma_2)
+\eta^{ab}(\pa_c V^c_1 \pa_b \sigma_2   -\pa_b V_1^c \pa_c  \sigma_2)
-\eta^{bc}  \pa_bV_1^a \pa_c \sigma_2
\right)
- 1\leftrightarrow 2
\cr
&= V_1^b \pa_b V_2^a
+\frac{\ell^4}{2r^2}e^{-2\varphi}
\eta^{ab} \pa_b (V_1^c \pa_c \sigma_2)
- 1\leftrightarrow 2\,,
\ee
where we used
\be
\eta^{ab}(\pa_c V^c_1 \pa_b \sigma_2   -\pa_b V_1^c \pa_c  \sigma_2)
-\eta^{bc}  \pa_bV_1^a \pa_c \sigma_2=0\,,
\ee
as it can be explicitly checked for the two different components $a=t,\phi$ using the holomorphicity conditions \eqref{VKill}.

\section{Charge computation} \label{CQ}

To begin, let us derive the relation
\be
  \partial_t ( V_\xi^a\pa_a \alpha)
   &= 
   \partial_\phi V_\xi^\phi\pa_t \alpha
   +\f1{\ell^2} V_\xi^t\pa^2_\phi \alpha
   +\f1{\ell^2} \partial_\phi V_\xi^t\pa_\phi \alpha
   + V_\xi^\phi\pa_\phi \partial_t \alpha
     \cr
   &=  \partial_\phi \left(
    V_\xi^\phi\pa_t \alpha
    +\f1{\ell^2} V_\xi^t\pa_\phi \alpha
   \right)\,,
   \label{key}
 \ee
 valid for a scalar function $\alpha$ such that $\Box \alpha=0$, and where we repeatedly used \eqref{VKill}. 
In particular, the relation \eqref{key} for $\alpha=\varphi$ implies 
 \be
 \Box (V_\xi^a\pa_a \varphi)=0\,,
 \label{boxV}
 \ee
 from which, using \eqref{chi}, \eqref{VKill}, \eqref{boxomega}, it follows that
 \be
 \Box\sigma_\xi=0\,.
 \la{boxsigma-app}
 \ee

The Einstein--Hilbert pre-symplectic potential takes the following expression
\be
 \bold{\theta}= \eps_\mu \theta^\mu\,,
\ee
where
\begin{equation}
\begin{aligned}
   \theta^\mu &= \f{1}{2}(g^{\alpha\beta} \delta \Gamma^\mu_{\alpha\beta}-g^{\alpha\mu} \delta \Gamma^\beta_{\alpha\beta})\\
    &=\f{1}{2}(g^{\mu\nu}\nabla^\sigma \var g_{\nu\sigma} - g^{\nu\sigma}\nabla^\mu \var g_{\nu\sigma}),
    \label{thetaEH}
\end{aligned}
\end{equation}and 
\be
\eps_\mu=n_\mu \teps=-s_\mu \beps\,,
\ee 
with the expressions for $\teps, \beps, s_\mu, n_\mu$  given in Section \ref{EH_sec}.  The Noether charge aspect is
\begin{equation}
    d\bold{q}_\xi= I_\xi\bold{\theta}- \iota_\xi\bold{L} - \bold{a}_\xi.
    \label{qxi}
\end{equation}

The field-contraction of the pre-symplectic potential with respect to the diffeomorphisms $\xi$ yields
\be
 I_\xi\bold{\theta} &=\f12 \eps_\mu (g^{\mu\nu}\nabla^\sigma \cL_\xi g_{\nu\sigma} - g^{\nu\sigma}\nabla^\mu  \cL_\xi g_{\nu\sigma})
 \cr
 &= \eps_\mu  g^{\mu\nu} (\nabla^\sigma \nabla_{(\nu} \xi_{\sigma)} -   \nabla_\sigma \nabla_\nu \xi^\sigma+ R_{\sigma\nu} \xi^\sigma)
  \cr
 &=\f12  \eps_\mu   \nabla_\sigma (\nabla^{\sigma} \xi^{\mu} - \nabla^{\mu} \xi^{\sigma})   - \frac{2}{\ell^2}   \eps_\mu \xi^\mu
  \cr
 &=- d\left( \f12  \bar \teps s_\sigma n_\mu  (\nabla^{\sigma} \xi^{\mu} - \nabla^{\mu} \xi^{\sigma}) \right)  - \frac{2}{\ell^2}   \eps_\mu \xi^\mu\,,
 \la{1noe}
\ee
where we used Einstein’s field equations in three dimensions to write $R_{\sigma\nu}=2\Lambda g_{\sigma\nu}$, and $\teps_\beta=- \bar \teps s_\beta$---which follows from $ \bar \teps=-\iota_{\hat s}\teps$.

Integrating the total derivative term in  \eqref{1noe} on the corner at infinity yields
\begin{equation}
\begin{aligned}
&-\frac{1}{2} \int_\cS \bar{\teps} \left(n_\mu s^\nu \nabla_\nu \xi^\mu -  n^\mu s_\nu\nabla_\mu \xi^\nu\right)\\
&=\frac{1}{2\ell^2}  \int_0^{2\pi} d\phi \  \Bigl(2\ell^2 V_\xi^\phi C^{(2)}_{t\phi} + \ell^2 V_\xi^t C_{tt}^{(2)} +V_\xi^t C_{\phi\phi}^{(2)} +\ell^4 \pa_t\sigma_\xi\Bigl)
\cr
&\quad +\frac{ r^2}{\ell^2} \int_0^{2\pi} d\phi \ e^{2\varphi} V_\xi^t,
\label{Xi}
\end{aligned}
\end{equation}
where we see  the last  term is divergent.  

The next ingredient we need is the anomaly of the EH Lagrangian $\bold{a}_\xi$. 
Since $I_{\delta\xi}\bold{L}_{\EH}=0$, this  is given by
\be
\Delta_\xi \bold{L}_{\EH}&=
\delta_\xi\bold{L}_{\EH} -\cL_\xi \bold{L}_{\EH}
\cr
&=
d \Bigl[\Bigl(\f{2r^2}{\ell^2}  e^{2\varphi} \var_\xi\varphi+\ell^2{ e^{2\varphi}}\Box \delta_\xi\varphi \ln r {+\cO(r^{-2})}\Bigl) dt\wedge d\phi\Bigl] -d\iota_\xi \bold{L}_{\EH}\,,
\ee
where we used \eqref{R0}.
From  the definition \eqref{DL},
it follows that 
\be
\bold{a}_\xi
&=
\Bigl(\f{2r^2}{\ell^2}  e^{2\varphi} \var_\xi\varphi {+\cO(r^{-2})}\Bigl) dt\wedge d\phi
+\f2{\ell^2} \xi^\mu \bold{\epsilon}_\mu
\,,
\label{Lanomaly-app}
\ee
where we used the condition $\Box \omega_\xi=0$.
Therefore, we see that the Lagrangian contribution in \eqref{qxi}
is such that 
\be
\iota_\xi \bold{L} &=  -\f{2}{\ell^2} \xi^\mu \eps_\mu
\cr
&=- \bold{a}_\xi\,
-\f{2r^2}{\ell^2}  e^{2\varphi} \omega_\xi dt\wedge d\phi\,.
\ee
This, together with the anomaly contributions in \eqref{qxi} and   the last contribution in \eqref{1noe}, yield
\be
-\iota_\xi \bold{L} - \bold{a}_\xi - \frac{2}{\ell^2}\xi^\mu \eps_\mu&=
\f{2r^2}{\ell^2}  e^{2\varphi} \omega_\xi dt\wedge d\phi- \frac{2}{\ell^2}\xi^\mu \eps_\mu
\cr
&=-\f{r^2}{\ell^2} \pa_a (e^{2\varphi} V_\xi^a) dt\wedge d\phi
\cr
&- \frac{1}{\ell^2} \pa_r \left(r^2  e^{2\varphi} V_\xi^t\right)  dr\wedge d\phi
+ \frac{1}{\ell^2} \pa_r \left(r^2  e^{2\varphi} V_\xi^\phi\right)   dr\wedge dt
\cr
&+{\ell^2}d \left(\ln{r}\pa_t\sigma_\xi  d\phi\right) +d \left(\ln{r} \pa_\phi \sigma_\xi dt \right)
\cr
&=
- \frac{1}{\ell^2}d \left(r^2  e^{2\varphi} V_\xi^t  d\phi\right)
+ \frac{1}{\ell^2}d \left(r^2  e^{2\varphi} V_\xi^\phi   dt\right)
\cr
&+{\ell^2}d \left(\ln{r}\pa_t\sigma_\xi  d\phi\right) +d \left(\ln{r} \pa_\phi \sigma_\xi dt \right)\,,
\ee
where we used \eqref{chi} and \eqref{boxsigma-app}.
For later use, we also notice that
\be
I_{{\delta \xi}}\bold{\theta}_{EH}
=  - \frac{2}{\ell^2} \eps_\mu\var \xi^\mu+\cO(r^{-2}).
\label{Ivarxi}
\ee

Therefore, the Noether charge reads
\begin{equation}
\begin{aligned}
\cQ_\xi  &= \frac{1}{2\ell^2}  \int_0^{2\pi}    \Bigl(2\ell^2 V_\xi^\phi C^{(2)}_{t\phi} + \ell^2 V_\xi^t C_{tt}^{(2)} +V_\xi^t C_{\phi\phi}^{(2)} +\ell^4 \pa_t\sigma_\xi\Bigl)\  d\phi
\cr
&+\int_0^{2\pi}  \ell^2\pa_t\sigma_\xi \ln r\ d\phi\,.
\la{Qxi-app}
\end{aligned}
\end{equation}
Using Fourier expansion of  the general solution of  \eqref{boxsigma-app}, it is immediate to see that the  partial-$t$ derivative terms in \eqref{Qxi-app} vanish, arriving at the final expression
\be
\cQ_\xi  &=\int_0^{2\pi} d\phi \   V^a_\xi C^{(2)}_{ta}\,,
\ee
where we used the condition \eqref{boxvarphi=0}.

\section{Flux computation} \label{CF}
 
We recall the expression of the Noetherian symplectic flux 
\be
    \cF_\xi &= \int_\cS (\iota_\xi \bold{\theta} +\bold{q}_{\var\xi} + \bold{A}_\xi)\,.
    \la{Fxi}
\ee

The different components of the pre-symplectic potential
\be
\theta^\mu_\EH= \f{1}{2} (g^{\alpha\beta} \delta \Gamma^\mu_{\alpha\beta}-g^{\alpha\mu} \delta \Gamma^\beta_{\alpha\beta})\,,
\ee
are given by (see Appendix \ref{Chris-app} for a list of Christoffel symbols)
\be
\theta^t_\EH
&=
 \f{\ell^2}{r^2}e^{-2\varphi} \delta \partial_t \varphi
+\cO(r^{-4})\,,
\label{thetat}
\\
\theta^r_\EH
&=
-\f{2r}{\ell^2}
 \delta \varphi
+\f{e^{-2\varphi}}{2r \ell^2}\delta(C_{\phi\phi}^{(2)}-\ell^2 C_{tt}^{(2)} )
+\cO(r^{-3})\,,
\label{thetar}
\\
\theta^\phi_\EH 
&=-\f{e^{-2\varphi}}{r^2} \delta \pa_\phi\varphi
+\cO(r^{-4})\,,
\label{thetap}
\ee
from which
\be
\bold{\theta}&=\d \bold{\eta}+\cO(r^{-2})\,,
\la{theta-app}
\ee
with
\be
\bold{\eta}=\left(
\f{r^2}{\ell^2}e^{2\varphi} \ dt\wedge d\phi +\f{\ell^2} r \pa_t\varphi \ dr \wedge d\phi +\f{1} r \pa_\phi \varphi \ dr \wedge dt
\right)
\ee

We first show that the  anomaly of the pre-symplectic potential \eqref{theta-app}  vanishes. Using the definition of the anomaly operator \eqref{ano}, we have
\be
\Delta_\xi \bold {\theta}
&=\d \d_\xi \bold {\eta}-\d d \iota_\xi  \bold {\eta} +  d \iota_{\d\xi}  \bold {\eta}- \iota_\xi \d d  \bold {\eta}-\d_{\d\xi}\bold {\eta}\,.
\la{Dxieta}
\ee
 Neglecting terms of $\cO(r^{-2})$  we have 
\be
(\Delta_\xi \bold {\theta})_{t\phi}&=
\d\left(-
\f{2r^2}{\ell^2}e^{2\varphi}\omega_\xi \right)\ dt\wedge d\phi 
\cr
&-d\left(
\f{2r^2}{\ell^2}e^{2\varphi}\d\varphi \left( \xi^t\ d\phi- \xi^\phi dt\right)
+\sigma_\xi (\ell^2 \pa_t \d\varphi  d\phi + \pa_\phi\d\varphi dt)
\right)
\cr
&- \iota_\xi \d \left(
\f{2r}{\ell^2}e^{2\varphi} \ dr\wedge dt\wedge d\phi 
\right)
\cr
&=-
\f{4r^2}{\ell^2}e^{2\varphi}\d\varphi (\omega_\xi +\sigma_\xi) \ dt\wedge d\phi
-
\f{2r^2}{\ell^2}e^{2\varphi}\d \omega_\xi \ dt\wedge d\phi 
\cr
&-\f{2r^2}{\ell^2}e^{2\varphi}
\left(
\xi^a \pa_a \d\varphi 
+2\xi^a \d \varphi \pa_a\varphi
+\d \varphi  \pa_a \xi^a
 \right)
\ dt\wedge d\phi 
\cr
&-\left( 
 \ell^2 \pa_t \sigma_\xi \pa_t \d\varphi  -\pa_\sigma \sigma_\xi \pa_\phi\d\varphi )
+\sigma_\xi (\ell^2 \pa^2_t  - \pa^2_\phi)\d\varphi
\right)
\ dt\wedge d\phi 
\cr
&=
-\f{2}{\ell^2}e^{2\varphi}
\left(
\xi^a_{(2)} \pa_a \d\varphi 
+2\xi^a_{(2)} \d \varphi \pa_a\varphi
+\d \varphi  \pa_a \xi^a_{(2)}
 \right)
\ dt\wedge d\phi 
\cr
&-
( \ell^2 \pa_t \sigma_\xi \pa_t \d\varphi  -\pa_\sigma \sigma_\xi \pa_\phi\d\varphi )
\ dt\wedge d\phi 
\cr
&=
-\f{2}{\ell^2}e^{2\varphi}
\left(
2\xi^a_{(2)} \d \varphi \pa_a\varphi
-2  \xi^a_{(2)}  \d \varphi \pa_a \varphi
+\f{\ell^4}{2r}\d \varphi \Box \sigma_\xi
 \right)
\ dt\wedge d\phi 
\cr
&=0\,,
\ee
\be
    (\Delta_\xi\bold{\theta})_{r\phi}&=-\frac{\ell^2}{r}\pa_t\d\omega_\xi dr\wedge d\phi-\iota_\xi\delta\left(-\f{2r}{\ell^2}e^{2\varphi}+\f{\ell^2}{r}\pa^2_t\varphi-\f1r\pa^2_\phi\varphi\right)dt\wedge dr\wedge d\phi
    \cr
    &-d\left(\left(2\d\varphi e^{2\varphi}\frac{r^2}{\ell^2}\xi^t+\ell^2\partial_t\d\varphi\sigma_\xi\right)d\phi-\left(\frac{\ell^2}{r}\partial_t\d\varphi\xi^\phi+\frac{1}{r}\pa_\phi\d\varphi \xi^t\right)dr\right)
    \cr
    &=\frac{\ell^2}{r}\pa_t(V_\xi^a\pa_a\d\varphi)dr\wedge d\phi+\frac{4r}{\ell^2}\d\varphi e^{2\varphi}V_{\xi}^t dr\wedge d\phi
    \cr
    &-\left(\frac{4r}{\ell^2}\d\varphi e^{2\varphi}V_\xi^t+\pa_{\phi}\left(\frac{1}{r}V_\xi^t\pa_\phi\d\varphi+\frac{\ell^2}{r} V_\xi^\phi\pa_t\d\varphi\right)\right)dr\wedge d\phi
    \cr
    &=  \frac{1}{r}\big(\ell^2\pa_t V_\xi^\phi\pa_\phi\d\varphi+\ell^2 V_\xi^\phi\pa_t\pa_\phi\d\varphi+\ell^2\pa_t V_\xi^t\pa_t\d\varphi +\ell^2 V_\xi^t\pa_t^2\d\varphi
    \cr
    &-\ell^2 V_\xi^\phi\pa_\phi\pa_t\d\varphi-\ell^2\pa_t\d\varphi\pa_\phi V_\xi^\phi-\pa_\phi^2\d\varphi V_\xi^t-\pa_\phi\d\varphi\pa_\phi V_{\xi}^t\big)dr \wedge d\phi
    \cr
    &=0\,,
\ee
\be
        (\Delta_{\xi}\bold{\theta})_{rt}&=-\f{1}{r}\partial_\phi\d\omega_\xi dr\wedge dt -\iota_\xi\d\left(-\f{2r}{\ell^2}e^{2\varphi}+\f{\ell^2}{r}\pa^2_t\varphi-\f1r\pa^2_\phi\varphi\right)dt\wedge dr\wedge d\phi
        \cr
        &-d\left(\left(-2\d\varphi\f{r^2}{\ell^2}e^{2\varphi}\xi^{\phi}+\sigma_\xi\partial_\phi\d\varphi\right)dt-\left(\frac{\ell^2}{r}\xi^{\phi}\partial_t\d\varphi+\f1r\xi^t\partial_{\phi}\d\varphi\right)dr\right)
        \cr
        &=\f1r\pa_\phi(V_\xi^a\pa_a\d\varphi)dr\wedge dt-4\d\varphi\frac{r}{\ell^2}e^{2\varphi}V_\xi^\phi dr\wedge dt
        \cr
        &+\left(4\delta\varphi\frac{r}{\ell^2}e^{2\varphi}V_\xi^\phi-\partial_t\left(\frac{\ell^2}{r}V_\xi^\phi\partial_t\d\varphi+\frac{1}{r}V_\xi^t\pa_\phi\d\varphi\right)\right)dr\wedge dt
        \cr
        &=\f1r\big(\pa_\phi V_\xi^\phi\pa_\phi\d\varphi+V_\xi^\phi\pa^2_\phi\d\varphi+\pa_\phi V_\xi^t\pa_t\d\varphi+V_\xi^t\pa_t\pa_\phi\d\varphi
        \cr
        &-\ell^2\pa_t^2\d\varphi V_\xi^\phi-\ell^2\partial_t\d\varphi\pa_t V_\xi^\phi-\pa_t\pa_\phi\d \varphi V^t_\xi-\pa_\phi\d\varphi\pa_t V_\xi^t\big)dr\wedge dt
        \cr
        &=0\,.
\ee
In the computation above we used \eqref{xia}, \eqref{chi}, \eqref{56},  \eqref{boxvarphi}, \eqref{boxs} and $\d\omega_\xi=-V^a_\xi\pa_a\d\varphi$.

Hence, using  $\Delta_\xi \bold {\theta}=0$, we have
\be
d(\bold{q}_{\var\xi} + \bold{A}_\xi)&=
I_{\var \xi}\bold{\theta} - \iota_{\var \xi}\bold{L} 
     -\var \bold{a}_\xi 
\cr
&=\f{2r}{\ell^2} \d(\sqrt{-g} \omega_\xi )dt\wedge d\phi
-\f2{\ell^2}\d( \xi^\mu  \eps_\mu)\,,
\ee
where we used \eqref{Ivarxi}.
Therefore, we have
\be
d(\bold{q}_{\var\xi} + \bold{A}_\xi)
&=\f{2r}{\ell^2} \d(\sqrt{-g} (\omega_\xi +\sigma_\xi))dt\wedge d\phi
-\f2{\ell^2}\d(\sqrt{-g}  ( \xi^t  dr\wedge d\phi-\xi^\phi  dr\wedge dt))
\cr
&=-\f{r^2}{\ell^2} \d\pa_a(e^{2\varphi} V^a_\xi)dt\wedge d\phi
-\f2{\ell^2}\d(\sqrt{-g}  ( V_\xi^t  dr\wedge d\phi-V_\xi^\phi  dr\wedge dt))
\cr
&+\d\left( \f{\ell^2}r \pa_t \sigma_\xi  dr\wedge d\phi+\f1r\pa_\phi \sigma_\xi dr\wedge dt\right)
\cr
&=-\d d\left( \f{r^2}{\ell^2} e^{2\varphi} V^t_\xi d\phi\right)
+\d d\left( \f{r^2}{\ell^2} e^{2\varphi} V^\phi_\xi dt\right)
\cr
&+\d d \left( \ell^2 \ln{r} \pa_t \sigma_\xi  d\phi+\ln{r}\pa_\phi \sigma_\xi dt\right)\,,
\ee
where again we used \eqref{boxsigma-app}.
We thus can write
\be
    \cF_\xi &= \int_\cS \iota_\xi \bold{\theta} 
    -\d \left( \int_\cS   \f{r^2}{\ell^2} e^{2\varphi} V^t_\xi d\phi \right)
  \,.
    \la{Fxi-2}
\ee

In order to compute the first term in \eqref{Fxi-2}, we use   the expression \eqref{theta-app} of the pre-symplectic potential components to derive
\be
 \cF^\theta_\xi&=
        \int_\cS \iota_\xi \bold{\theta}= \int_0^{2\pi}d\phi\ \sqrt{-g} \ (  \xi^r \theta^t-\xi^t \theta^r) 
\cr
 &=\int_0^{2\pi}d\phi\ 
 \bigg[
\ell^2 (\sigma_\xi\delta \partial_t \varphi
- \partial_t \sigma_\xi \delta \varphi) \cr
& +\f{1}{\ell^2}
\left(2r^2e^{2\varphi}V_\xi^t
+ V_\xi^t \left(C_{\phi\phi}^{(2)}-\ell^2 C_{tt}^{(2)} \right)
 \right)
 \delta \varphi
-\f{V_\xi^t}{2 \ell^2} \delta\left(C_{\phi\phi}^{(2)}-\ell^2 C_{tt}^{(2)} \right)
\bigg]\,.
\label{Fthetaxi}
\ee
We now use the on-shell relation
\be
C_{\phi\phi}^{(2)} -\ell^2 C_{tt}^{(2)}
=\ell^4 e^{2\varphi}\Box \varphi
\ee
to write 
\be
\left(C_{\phi\phi}^{(2)}-\ell^2 C_{tt}^{(2)} \right)
 \delta \varphi-
 \f12\delta\left(C_{\phi\phi}^{(2)}-\ell^2 C_{tt}^{(2)} \right)
 &=\f{\ell^4}2 \Box \varphi\delta e^{2\varphi}-
 \f{\ell^4}2\delta\left(e^{2\varphi}\Box \varphi \right)
 \cr
& =-
 \f{\ell^4}2 e^{2\varphi} \delta\left(\Box \varphi \right)\,,
\ee
from which
\be
 \cF^\theta_\xi&\heq
 \delta\left(
\f{r^2}{\ell^2}\int_0^{2\pi}d\phi\ e^{2\varphi}V_\xi^t
+\ell^2 \int_0^{2\pi}d\phi\ 
( \sigma_\xi \partial_t \varphi
-\varphi \partial_t \sigma_\xi)
\right)
-\frac{\ell^2}{2} \int_0^{2\pi}d\phi\ V_{\xi}^te^{2\varphi} \delta\left(\Box \varphi \right)\,.
 \ee
We thus see that the boundary condition \eqref{boxvarphi=0} is crucial in order to have an integrable symplectic flux and finally arrive at
\be
\cF_\xi= \delta\left(
\ell^2 \int_0^{2\pi}d\phi\ 
( \sigma_\xi \partial_t \varphi
-\varphi \partial_t \sigma_\xi)
\right)\,.
\ee

\section{Charge bracket}\label{Qalgebra-app}

By means of 
\be
\delta_\xi C^{(2)}_{ab} 
&=  \cL_{V_\xi} C_{ab}^{(2)} +\ell^2\left(\eta_{ab}\eta^{cd}\pa_c\sigma_\xi \pa_d \varphi +  \pa_b\pa_a\sigma_\xi -2\pa_{(a}\varphi \pa_{b)}\sigma_\xi\right)\,,
\ee
we can compute the Noether charge (field-space) Lie derivative
\be
\var_{\xi_2} \cQ^{\van EH}_{\xi_1}&= 
\var_{\xi_2} \left(\int_0^{2\pi} d\phi \   V^a_1 C^{(2)}_{ta}  \right)
\cr
&= \ell^2 \int_0^{2\pi} d\phi \  \bigg[\f1{\ell^2}  V^a_1 \cL_2C^{(2)}_{ta}  + V^a_1 (\eta_{at}\eta^{cd}\pa_c\sigma_2 \pa_d \varphi +  \pa_t\pa_a\sigma_2 -2\pa_{(a}\varphi \pa_{t)}\sigma_2) \bigg]
\cr
&=\ell^2\int_0^{2\pi} d\phi \  \bigg[\f1{\ell^2}  V^a_1 \cL_2C^{(2)}_{ta} 
-V^a_1\pa_{a}\sigma_2\pa_{t}\varphi 
\cr
& +\varphi\pa_\phi
\left(
V^\phi_1\pa_{t}\sigma_2
+\f1{\ell^2} V^t_1\pa_\phi\sigma_2
\right)
 +  V^a_1\pa_t\pa_a\sigma_2 
\bigg]
\cr
&=\ell^2\int_0^{2\pi} d\phi \  \bigg[\f1{\ell^2}  V^a_1 \cL_2C^{(2)}_{ta} 
-V^a_1\pa_{a}\sigma_2\pa_{t}\varphi 
\cr
& +\varphi
\left(
\pa_\phi V^\phi_1\pa_{t}\sigma_2
+V^\phi_1\pa_{t}\pa_\phi\sigma_2
+\f1{\ell^2}\pa_\phi V^t_1\pa_\phi\sigma_2
+\f1{\ell^2} V^t_1\pa^2_\phi\sigma_2
\right)
 +  V^a_1\pa_t\pa_a\sigma_2 
\bigg]
\cr
&=\ell^2\int_0^{2\pi} d\phi \  \bigg[-\f1{\ell^2} \cL_2 V^a_1 C^{(2)}_{ta} 
-V^a_1\pa_{a}\sigma_2\pa_{t}\varphi 
 +\varphi
\pa_{t} \left(
V^a_1\pa_a\sigma_2
\right)
 +  V^a_1\pa_t\pa_a\sigma_2 
\bigg]\,,
\ee
where we used \eqref{VKill}, \eqref{boxsigma-app},
and that  on-shell
\be
 \int_0^{2\pi} d\phi\ \cL_2 (V^a_1 C^{(2)}_{ta} )=0\,.
 \ee
Next, we have
\be
-\var_{\xi_2} \cP^{\van EH}_{\xi_1}
&=\ell^2 \int_0^{2\pi}d\phi\ ( \sigma_1 \partial_t \omega_2 -\omega_2 \partial_t \sigma_1)
\cr
&=-\ell^2 \int_0^{2\pi}d\phi\ \bigg[ \sigma_1 \partial_t \sigma_2 -\sigma_2 \partial_t \sigma_1
+\f12
\sigma_1 \partial_t \pa_a V^a_2
-\f12 \pa_a V^a_2 \partial_t \sigma_1
\cr
&+\sigma_1 \partial_t ( V^a_2\pa_a \varphi)
-V^a_2\pa_a \varphi \partial_t \sigma_1
\bigg]\,,
\ee
where we used  \eqref{chi}.
We now massage  the  $\varphi$-dependent terms  in the last expression;
by means of the relation \eqref{key},
we can write
 \be
 &-\int_0^{2\pi}d\phi\ \bigg[
\sigma_1 \partial_t ( V^a_2\pa_a \varphi)
-V^a_2\pa_a \varphi \partial_t \sigma_1
 \bigg]
 \cr
 &=\int_0^{2\pi}d\phi\ \bigg[
 V^a_2\pa_a \sigma_1 \partial_t \varphi
+  \f1{\ell^2}V^t_2 \pa_\phi\sigma_1\pa_\phi \varphi
+V^\phi_2\pa_\phi \varphi \partial_t \sigma_1
 \bigg]
\cr
&=
 \int_0^{2\pi}d\phi\ \bigg[
V^a_2\pa_a \sigma_1 \partial_t \varphi
-\varphi
\pa_{t} \left(
V^a_2\pa_a\sigma_1
\right)\bigg]\,,
\label{PP1}
 \ee
 where again we used \eqref{boxsigma-app}.
Moreover,   we can use \eqref{key} to write
 \be
 \f12 \int_0^{2\pi}d\phi\ \left(
 \sigma_1 \partial_t \pa_a V^a_2
- \pa_a V^a_2 \partial_t \sigma_1\right)
=\int_0^{2\pi}d\phi\ V^a_2\pa_t\pa_a\sigma_1\,.
\label{PP2}
 \ee

 Relations \eqref{PP1}, \eqref{PP2} allow us to write
 \be
 -\var_{\xi_2} \cP_{\xi_1}&= 
 \ell^2\int_0^{2\pi}d\phi\ \left[
V^a_2\pa_a \sigma_1 \partial_t \varphi
-\varphi
\pa_{t} \left(
V^a_2\pa_a\sigma_1
\right)\right]
\cr
&-
 \ell^2 \int_0^{2\pi}d\phi\ \left[ \sigma_1 \partial_t \sigma_2 -\sigma_2 \partial_t \sigma_1
+V^a_2\pa_t\pa_a\sigma_1\right]\,.
 \ee

\section{Cocycle}\label{L_an}

In order to compute the cocycle \eqref{cocy}, we need to derive the 1-form $\bold{c}_{(\xi,\zeta)}$ given by
\be
 d\bold{c}_{(\xi,\zeta)} = (\delta_\xi\bold{a}_\zeta -\cL_\xi\bold{a}_\zeta -\xi\leftrightarrow \zeta) +\bold{a}_{\lbr \xi, \zeta \rbr }\,.
\ee
In all the calculations that follow, we will neglect sub-leading terms of order $\cO(r^{-2})$ or higher, as those won't contribute when evaluated at the corner at infinity. 
We start with the Lie derivative  of \eqref{Lanomaly}, namely $\cL_\xi\bold{a}_\zeta= d \iota_\xi \bold{a}_\zeta +\iota_\xi d\bold{a}_\zeta $, and focus on the second piece which is not an explicit total derivative; this gives
\be
\iota_\xi d \bold{a}_\zeta&=\f12 \iota_\xi(\pa_\sigma \bold{a}_{\zeta \mu \nu} dx^\sigma \wedge dx^\mu\wedge dx^\nu)
\cr
&=\f12 \left(\xi^\sigma\pa_\sigma \bold{a}_{\zeta \mu \nu}  dx^\mu\wedge dx^\nu
-2\xi^\nu \pa_\sigma \bold{a}_{\zeta \mu \nu}   dx^\mu \wedge dx^\sigma
\right)
\cr
&=\left(
 \xi^r\pa_r \bold{a}_{\zeta t \phi} 
- \xi^r \pa_\phi \bold{a}_{\zeta t r}
   -\xi^r \pa_t \bold{a}_{\zeta  r \phi } 
   \right) dt\wedge d\phi
\cr
&+\left(
\xi^\phi\pa_\phi \bold{a}_{\zeta t r} 
 - \xi^\phi\pa_r \bold{a}_{\zeta t \phi}
 +\xi^\phi\pa_t \bold{a}_{\zeta r \phi}  
  \right) dt\wedge dr
\cr
&+\left(
\xi^t\pa_t \bold{a}_{\zeta r \phi} 
+ \xi^t \pa_\phi \bold{a}_{\zeta  tr }
 -\xi^t\pa_r \bold{a}_{\zeta  t\phi } 
\right)  dr\wedge d\phi\,,
\ee
from which
\be
\iota_\xi d \bold{a}_\zeta-\xi\leftrightarrow \zeta
&=\f2{\ell^2}\left(
- 2 \sqrt{-g}\xi^r (\omega_\zeta+\sigma_\zeta)
- \xi^r \pa_\phi (\sqrt{-g} \zeta^\phi)
   -\xi^r \pa_t  (\sqrt{-g} \zeta^t)
   \right) dt\wedge d\phi
   \cr
&+\f2{\ell^2} \left(
\xi^\phi\pa_\phi (\sqrt{-g} \zeta^\phi)
 +2\sqrt{-g} \xi^\phi (\omega_\zeta+\sigma_\zeta)
 +\xi^\phi\pa_t (\sqrt{-g} \zeta^t)
  \right) dt\wedge dr
  \cr
&+\f2{\ell^2}\left(
\xi^t\pa_t(\sqrt{-g} \zeta^t)
+ \xi^t \pa_\phi  (\sqrt{-g} \zeta^\phi)
 +2 \sqrt{-g}\xi^t  (\omega_\zeta+\sigma_\zeta)
\right)  dr\wedge d\phi\
 \cr
&
  -\xi\leftrightarrow \zeta
\cr
&=   \f2{\ell^2}\left( \xi^\mu \pa_a ( \epsilon_\mu \zeta^a)
+2   (\omega_\zeta+\sigma_\zeta) \xi^\mu \epsilon_\mu
\right)
  -\xi\leftrightarrow \zeta
  \cr
&=   \f2{\ell^2}\left( \xi^\mu \pa_\nu ( \epsilon_\mu \zeta^\nu)
+2   \omega_\zeta \xi^\mu \epsilon_\mu
\right)
  -\xi\leftrightarrow \zeta\,,
  \la{iotada}
\ee
where in the last passage we used that, on-shell of $\Box\varphi=0$, we have
\be
\sqrt{-g}&=re^{2\varphi}+ \cO(r^{-3})\,.
\ee

Next, we compute
\be
 \d_\xi \bold{a}_\zeta-\xi\leftrightarrow \zeta
&=
 \f2{\ell^2}
 \left( \d_\xi  \zeta^\mu \epsilon_\mu 
 -2 \omega_\xi  \zeta^\mu \epsilon_\mu 
+ r^2  e^{2\varphi} \var_\xi \var_\zeta \varphi \,dt\wedge d\phi
\right)-\xi\leftrightarrow \zeta\,,
 \ee
 where we used 
$\var_\xi\varphi=-\omega_\xi$. 
Moreover, we have
\be
 \bold{a}_{\lbr \xi, \zeta \rbr}&=
\f2{\ell^2}
\left(
\lbr \xi, \zeta\rbr^\mu \epsilon_\mu
+ r^2 e^{2\varphi}  \delta_{\lbr \xi, \zeta \rbr} \varphi\ dt\wedge d\phi
\right)
\cr
&=\f2{\ell^2}
\left(
  (\xi^\nu\partial_\nu \zeta^\mu +\delta_\zeta \xi^\mu)\,\epsilon_\mu
  + r^2 e^{2\varphi}  \delta_\zeta   \delta_\xi\varphi\ dt\wedge d\phi
  \right)
 -\xi\leftrightarrow \zeta\,.
\ee
We thus can write
\be
 (\delta_\xi\bold{a}_\zeta  -\xi\leftrightarrow \zeta) +\bold{a}_{\lbr \xi, \zeta \rbr }
 &=\f2{\ell^2}
\left(
  \xi^\nu\partial_\nu \zeta^\mu -2 \omega_\xi  \zeta^\mu 
  \right)\epsilon_\mu 
  -\xi\leftrightarrow \zeta\,.
  \la{dabra}
\ee

Using \eqref{iotada} and \eqref{dabra}, we  obtain
\be
 d\bold{c}_{(\xi,\zeta)}
 &=-d\iota_\xi \bold{a}_\zeta
  +\f2{\ell^2}\partial_\nu \left( 
   \epsilon_\mu  \xi^\nu \zeta^\mu
\right)
  -\xi\leftrightarrow \zeta\,.
 \ee
 We now express
 \be
 \partial_\nu \left( 
   \epsilon_\mu  \xi^\nu \zeta^\mu
\right)
  -\xi\leftrightarrow \zeta
  &=
   \partial_r \left( 
   \epsilon_a  \xi^r \zeta^a
\right)
+\partial_a \left( 
   \epsilon_r  \xi^a \zeta^r
\right)
+\partial_a \left( 
   \epsilon_b  \xi^a \zeta^b
\right)
  -\xi\leftrightarrow \zeta
   \cr
  &=d \left( 
  \sqrt{-g}   \xi^r V_\zeta^a \varepsilon_{ab}\  dx^b
\right)
+d\left(\sqrt{-g}\xi^t\zeta^\phi\ dr \right)
  -\xi\leftrightarrow \zeta\,,
  \la{pxizeta}
 \ee
 where $\varepsilon_{t\phi}=1$, and we used
 \be
 \partial_a \left( 
   \epsilon_b  \xi^a \zeta^b
\right)
  -\xi\leftrightarrow \zeta=d\left(\sqrt{-g}\xi^t\zeta^\phi\ dr \right)
  -\xi\leftrightarrow \zeta\,,
 \ee
 and
 \be
 \partial_a \left(  \sqrt{-g}
 \xi^r \zeta^a \ dt\wedge d\phi
\right)
 -\xi\leftrightarrow \zeta
&=
\partial_a \left(  \sqrt{-g}
 \xi^r V_\zeta^a \ dt\wedge d\phi
\right)
+\frac{\ell^4}{2}\partial_a \left(  
 \sigma_\xi  \eta^{ab}\pa_b \sigma_\zeta\ dt\wedge d\phi
\right)
 -\xi\leftrightarrow \zeta
 \cr
 &= \partial_c \left(  \sqrt{-g}
 \xi^r V_\zeta^a \varepsilon_{ab} \ dx^c\wedge dx^b
\right)
 -\xi\leftrightarrow \zeta\,.
 \ee
 We can thus express
 \be
 d\bold{c}_{(\xi,\zeta)}& = 
-d\iota_\xi \bold{a}_\zeta
  + \frac{2}{\ell^2}d \left( 
  \sqrt{-g}   \xi^r V_\zeta^a \varepsilon_{ab}\  dx^b
\right)
+\frac{2}{\ell^2}d\left(\sqrt{-g}\xi^t\zeta^\phi\ dr \right) -\xi\leftrightarrow \zeta\,.
\la{cxizeta}
 \ee

The first  contribution to the cocyle \eqref{cocy} reads
\be
\int_\cS \iota_\xi \iota_\zeta \bold{L}_{\EH}&=
-\f{2}{\ell^2}\int_\cS \iota_\xi \iota_\zeta \eps
\cr
&=-\f{2}{\ell^2}\int_0^{2\pi} (\zeta^t\xi^r-\zeta^r\xi^t)\sqrt{-g} \,d\phi\,.
\la{iiL}
\ee
From \eqref{cxizeta} and \eqref{iiL}, we finally obtain
 \be
 \cK^L_{(\xi, \zeta)}&=\int_{\cS}\iota_\xi\iota_\zeta \bold{L}_{\EH} + \int_\cS(\iota_\xi \bold{a}_\zeta-\iota_\zeta \bold{a}_\xi)+\int_\cS \bold{c}_{(\xi,\zeta)} 
 \cr
&= \ell^2 \int_0^{2\pi}  \sigma_\xi \pa_t \sigma_\zeta \, d\phi
 -\xi\leftrightarrow \zeta\,.
 \la{Kxizeta}
\ee

\section{List of Christoffel symbols}\la{Chris-app}

 \begin{equation}
\begin{aligned}
 \Gamma^r_{rr} &= -\frac{1}{r}\,,\\  
\Gamma^r_{tt} &= \f{r^3}{\ell^4} e^{2\varphi} + \f{1}{r}  C_{tt}^{(4)} +  \cO(r^{-3})\,, \\
\Gamma^r_{t\phi}& = \frac{1}{r} C_{t\phi}^{(4)}+\cO(r^{-3})\,,\\ 
\Gamma^r_{\phi\phi}& =  - \f{r^3}{\ell^2} e^{2\varphi} + \f{1}{r} C_{\phi\phi}^{(4)}+\cO(r^{-3})\,,\\
\Gamma^t_{rt} &=\frac{1}{r} + \frac{1}{r^3} \ell^2C_{tt}^{(2)}e^{-2\varphi} +\cO(r^{-5})\,,\\
\Gamma^t_{r\phi}& = \frac{\ell^2}{r^3}e^{-2\varphi}{C_{t\phi}^{(2)}} +\cO(r^{-5})\,,\\
{\color{blue}}
\Gamma^t_{tt} &=\Bigl(1 + \frac{1}{r^2} \ell^2C_{tt}^{(2)}e^{-2\varphi} \Bigl) \pa_t \varphi - \frac{\ell^2}{2r^2} e^{-2\varphi}\pa_t C_{tt}^{(2)} +\frac{1}{r^2}e^{-2\varphi} C_{t\phi}^{(2)} \pa_\phi \varphi +\cO(r^{-4}) \,,\\
\Gamma^t_{t\phi} &=  \Bigl(1 + \frac{1}{r^2} e^{-2\varphi}\ell^2C_{tt}^{(2)} \Bigl)\pa_\phi \varphi - \frac{\ell^2}{2r^2}e^{-2\varphi} \pa_\phi C_{tt}^{(2)} + \frac{\ell^2}{r^2} e^{-2\varphi} C_{t\phi}^{(2)} \pa_t \varphi +\cO(r^{-4})\,, \\
\Gamma^t_{\phi \phi} &= \Bigl(\ell^2 + \frac{\ell^4}{r^2}e^{-2\varphi}C_{tt}^{(2)}\Bigl)\pa_t \varphi+ \frac{\ell^2}{2r^2} e^{-2\varphi} \pa_t C_{\phi\phi}^{(2)} - \frac{\ell^2}{r^2} e^{-2\varphi} \pa_\phi C_{t\phi}^{(2)} + \frac{\ell^2}{r^2}  e^{-2\varphi}C_{t\phi}^{(2)} \pa_\phi \varphi +\cO(r^{-4})\,, \\
\Gamma^\phi_{rt} &= - \frac{1}{r^3} e^{-2\varphi}C_{t\phi}^{(2)} +\cO(r^{-5})\,,\\
\Gamma^\phi_{r\phi} &= \frac{1}{r} - \frac{e^{-2\varphi}}{r^3} C_{\phi\phi}^{(2)}+\cO(r^{-5})\,,\\
\Gamma^\phi_{tt} &=  \Bigl(\frac{1}{\ell^2} - \frac{1}{r^2}C_{tt}^{(2)} e^{-2\varphi}\Bigl) \pa_\phi \varphi  - \frac{1}{r^2} e^{-2\varphi} C_{t\phi}^{(2)}  \pa_t \varphi  - \frac{1}{2r^2} e^{-2\varphi}\pa_\phi C_{tt}^{(2)} +\frac{1}{r^2} e^{-2\varphi} \pa_t C_{t\phi}^{(2)} +\cO(r^{-4})\,,\\
\Gamma^\phi_{t\phi} &=  \Bigl(1 -\frac{1}{r^2}C_{\phi\phi}^{(2)}e^{-2\varphi}\Bigl)\pa_t \varphi +\frac{1}{2r^2} e^{-2\varphi} \pa_t C_{\phi\phi}^{(2)} - \frac{1}{r^2} C_{t\phi}^{(2)}e^{-2\varphi}\pa_\phi \varphi +\cO(r^{-4})\,,\\
\Gamma^\phi_{\phi\phi}& = \Bigl(1 - \frac{1}{r^2}C_{\phi\phi}^{(2)}e^{-2\varphi}\Bigl) \pa_\phi \varphi - \frac{\ell^2}{r^2} e^{-2\varphi} C_{t\phi}^{(2)} \pa_t \varphi +\frac{1}{2r^2} e^{-2\varphi} \pa_\phi C_{\phi\phi}^{(2)} +\cO(r^{-4})\,.
\end{aligned}
\end{equation}

\bibliographystyle{bib-style2.bst}
\bibliography{biblio}

@article{Brown:1986nw,
    author = "Brown, J. David and Henneaux, M.",
    title = "{Central Charges in the Canonical Realization of Asymptotic Symmetries: An Example from Three-Dimensional Gravity}",
    doi = "10.1007/BF01211590",
    journal = "Commun. Math. Phys.",
    volume = "104",
    pages = "207--226",
    year = "1986"
}

@article{Barnich:2001jy,
    author = "Barnich, Glenn and Brandt, Friedemann",
    title = "{Covariant theory of asymptotic symmetries, conservation laws and central charges}",
    eprint = "hep-th/0111246",
    archivePrefix = "arXiv",
    reportNumber = "ULB-TH-01-19, MPI-MIS-94-2001",
    doi = "10.1016/S0550-3213(02)00251-1",
    journal = "Nucl. Phys. B",
    volume = "633",
    pages = "3--82",
    year = "2002"
}

@article{Barnich:2003xg,
    author = "Barnich, Glenn",
    title = "{Boundary charges in gauge theories: Using Stokes theorem in the bulk}",
    eprint = "hep-th/0301039",
    archivePrefix = "arXiv",
    reportNumber = "ULB-TH-03-01",
    doi = "10.1088/0264-9381/20/16/310",
    journal = "Class. Quant. Grav.",
    volume = "20",
    pages = "3685--3698",
    year = "2003"
}

@article{Barnich:2007bf,
    author = "Barnich, Glenn and Compere, Geoffrey",
    title = "{Surface charge algebra in gauge theories and thermodynamic integrability}",
    eprint = "0708.2378",
    archivePrefix = "arXiv",
    primaryClass = "gr-qc",
    reportNumber = "ULB-TH-06-30",
    doi = "10.1063/1.2889721",
    journal = "J. Math. Phys.",
    volume = "49",
    pages = "042901",
    year = "2008"
}

@article{Iyer:1994ys,
    author = "Iyer, Vivek and Wald, Robert M.",
    title = "{Some properties of Noether charge and a proposal for dynamical black hole entropy}",
    eprint = "gr-qc/9403028",
    archivePrefix = "arXiv",
    doi = "10.1103/PhysRevD.50.846",
    journal = "Phys. Rev. D",
    volume = "50",
    pages = "846--864",
    year = "1994"
}

@article{Lee:1990nz,
	author = {Lee, J. and Wald, Robert M.},
	doi = {10.1063/1.528801},
	journal = {J. Math. Phys.},
	pages = {725--743},
	title = {{Local symmetries and constraints}},
	volume = {31},
	year = {1990}}

@article{anderson1992introduction,
	author = {Anderson, Ian M},
	journal = {Contemporary Mathematics},
	title = {Introduction to the variational bicomplex},
	volume = {132},
	year = {1992}}

@article{Henningson:1998gx,
    author = "Henningson, M. and Skenderis, K.",
    title = "{The Holographic Weyl anomaly}",
    eprint = "hep-th/9806087",
    archivePrefix = "arXiv",
    reportNumber = "CERN-TH-98-188, KUL-TF-98-21",
    doi = "10.1088/1126-6708/1998/07/023",
    journal = "JHEP",
    volume = "07",
    pages = "023",
    year = "1998"
}

@article{Compere:2020lrt,
    author = "Comp{\`e}re, Geoffrey and Fiorucci, Adrien and Ruzziconi, Romain",
    title = "{The $\Lambda$-BMS$_4$ charge algebra}",
    eprint = "2004.10769",
    archivePrefix = "arXiv",
    primaryClass = "hep-th",
    doi = "10.1007/JHEP10(2020)205",
    journal = "JHEP",
    volume = "10",
    pages = "205",
    year = "2020"
}

@article{Barnich:2010eb,
    author = "Barnich, Glenn and Troessaert, Cedric",
    title = "{Aspects of the BMS/CFT correspondence}",
    eprint = "1001.1541",
    archivePrefix = "arXiv",
    primaryClass = "hep-th",
    reportNumber = "ULB-TH-09-28",
    doi = "10.1007/JHEP05(2010)062",
    journal = "JHEP",
    volume = "05",
    pages = "062",
    year = "2010"
}

@article{Freidel:2021cjp,
    author = "Freidel, Laurent and Oliveri, Roberto and Pranzetti, Daniele and Speziale, Simone",
    title = "{Extended corner symmetry, charge bracket and Einstein\textquoteright{}s equations}",
    eprint = "2104.12881",
    archivePrefix = "arXiv",
    primaryClass = "hep-th",
    doi = "10.1007/JHEP09(2021)083",
    journal = "JHEP",
    volume = "09",
    pages = "083",
    year = "2021"
}

@article{Donnelly:2016auv,
    author = "Donnelly, William and Freidel, Laurent",
    title = "{Local subsystems in gauge theory and gravity}",
    eprint = "1601.04744",
    archivePrefix = "arXiv",
    primaryClass = "hep-th",
    doi = "10.1007/JHEP09(2016)102",
    journal = "JHEP",
    volume = "09",
    pages = "102",
    year = "2016"
}

@article{Freidel:2020xyx,
    author = "Freidel, Laurent and Geiller, Marc and Pranzetti, Daniele",
    title = "{Edge modes of gravity. Part I. Corner potentials and charges}",
    eprint = "2006.12527",
    archivePrefix = "arXiv",
    primaryClass = "hep-th",
    doi = "10.1007/JHEP11(2020)026",
    journal = "JHEP",
    volume = "11",
    pages = "026",
    year = "2020"
}

@article{Freidel:2021fxf,
    author = "Freidel, Laurent and Oliveri, Roberto and Pranzetti, Daniele and Speziale, Simone",
    title = "{The Weyl BMS group and Einstein{\textquoteright}s equations}",
    eprint = "2104.05793",
    archivePrefix = "arXiv",
    primaryClass = "hep-th",
    doi = "10.1007/JHEP07(2021)170",
    journal = "JHEP",
    volume = "07",
    pages = "170",
    year = "2021"
}

@article{Compere:2018aar,
    author = "Comp\`ere, Geoffrey and Fiorucci, Adrien",
    title = "{Advanced Lectures on General Relativity}",
    eprint = "1801.07064",
    archivePrefix = "arXiv",
    primaryClass = "hep-th",
    month = "1",
    year = "2018"
}

@article{Chandrasekaran:2020wwn,
    author = "Chandrasekaran, Venkatesa and Speranza, Antony J.",
    title = "{Anomalies in gravitational charge algebras of null boundaries and black hole entropy}",
    eprint = "2009.10739",
    archivePrefix = "arXiv",
    primaryClass = "hep-th",
    doi = "10.1007/JHEP01(2021)137",
    journal = "JHEP",
    volume = "01",
    pages = "137",
    year = "2021"
}

@article{Ashtekar:1981bq,
    author = "Ashtekar, A. and Streubel, M.",
    title = "{Symplectic Geometry of Radiative Modes and Conserved Quantities at Null Infinity}",
    doi = "10.1098/rspa.1981.0109",
    journal = "Proc. Roy. Soc. Lond. A",
    volume = "376",
    pages = "585--607",
    year = "1981"
}

@article{Henneaux:2019sjx,
    author = "Henneaux, Marc and Merbis, Wout and Ranjbar, Arash",
    title = "{Asymptotic dynamics of AdS$_3$ gravity with two asymptotic regions}",
    eprint = "1912.09465",
    archivePrefix = "arXiv",
    primaryClass = "hep-th",
    doi = "10.1007/JHEP03(2020)064",
    journal = "JHEP",
    volume = "03",
    pages = "064",
    year = "2020"
}

@article{Margalef-Bentabol:2020teu,
    author = "Margalef-Bentabol, Juan and Villase{\~n}or, Eduardo J. S.",
    title = "{Geometric formulation of the Covariant Phase Space methods with boundaries}",
    eprint = "2008.01842",
    archivePrefix = "arXiv",
    primaryClass = "math-ph",
    doi = "10.1103/PhysRevD.103.025011",
    journal = "Phys. Rev. D",
    volume = "103",
    number = "2",
    pages = "025011",
    year = "2021"
}

@article{Wieland:2021eth,
    author = "Wieland, Wolfgang",
    title = "{Barnich{\textendash}Troessaert bracket as a Dirac bracket on the covariant phase space}",
    eprint = "2104.08377",
    archivePrefix = "arXiv",
    primaryClass = "hep-th",
    doi = "10.1088/1361-6382/ac3e52",
    journal = "Class. Quant. Grav.",
    volume = "39",
    number = "2",
    pages = "025016",
    year = "2022"
}

@article{Ciambelli:2021vnn,
    author = "Ciambelli, Luca and Leigh, Robert G.",
    title = "{Isolated surfaces and symmetries of gravity}",
    eprint = "2104.07643",
    archivePrefix = "arXiv",
    primaryClass = "hep-th",
    doi = "10.1103/PhysRevD.104.046005",
    journal = "Phys. Rev. D",
    volume = "104",
    number = "4",
    pages = "046005",
    year = "2021"
}

@article{Carrozza:2021gju,
    author = "Carrozza, Sylvain and Hoehn, Philipp A.",
    title = "{Edge modes as reference frames and boundary actions from post-selection}",
    eprint = "2109.06184",
    archivePrefix = "arXiv",
    primaryClass = "hep-th",
    doi = "10.1007/JHEP02(2022)172",
    journal = "JHEP",
    volume = "02",
    pages = "172",
    year = "2022"
}

@article{Adami:2021nnf,
    author = "Adami, H. and Grumiller, D. and Sheikh-Jabbari, M. M. and Taghiloo, V. and Yavartanoo, H. and Zwikel, C.",
    title = "{Null boundary phase space: slicings, news {\&} memory}",
    eprint = "2110.04218",
    archivePrefix = "arXiv",
    primaryClass = "hep-th",
    doi = "10.1007/JHEP11(2021)155",
    journal = "JHEP",
    volume = "11",
    pages = "155",
    year = "2021"
}

@article{Chandrasekaran:2021vyu,
    author = "Chandrasekaran, Venkatesa and Flanagan, Eanna E. and Shehzad, Ibrahim and Speranza, Antony J.",
    title = "{A general framework for gravitational charges and holographic renormalization}",
    eprint = "2111.11974",
    archivePrefix = "arXiv",
    primaryClass = "gr-qc",
    doi = "10.1142/S0217751X22501056",
    journal = "Int. J. Mod. Phys. A",
    volume = "37",
    number = "17",
    pages = "2250105",
    year = "2022"
}

@article{Freidel:2021dxw,
    author = "Freidel, Laurent",
    title = "{A canonical bracket for open gravitational system}",
    eprint = "2111.14747",
    archivePrefix = "arXiv",
    primaryClass = "hep-th",
    month = "11",
    year = "2021"
}

@article{Carrozza:2022xut,
    author = "Carrozza, Sylvain and Eccles, Stefan and Hoehn, Philipp A.",
    title = "{Edge modes as dynamical frames: charges from post-selection in generally covariant theories}",
    eprint = "2205.00913",
    archivePrefix = "arXiv",
    primaryClass = "hep-th",
    doi = "10.21468/SciPostPhys.17.2.048",
    journal = "SciPost Phys.",
    volume = "17",
    number = "2",
    pages = "048",
    year = "2024"
}

@article{Nagy:2022xxs,
    author = "Nagy, Silvia and Peraza, Javier",
    title = "{Radiative phase space extensions at all orders in r for self-dual Yang-Mills and gravity}",
    eprint = "2211.12991",
    archivePrefix = "arXiv",
    primaryClass = "hep-th",
    doi = "10.1007/JHEP02(2023)202",
    journal = "JHEP",
    volume = "02",
    pages = "202",
    year = "2023"
}

@article{Kabel:2023jve,
    author = "Kabel, Viktoria and Brukner, {\v{C}}aslav and Wieland, Wolfgang",
    title = "{Quantum reference frames at the boundary of spacetime}",
    eprint = "2302.11629",
    archivePrefix = "arXiv",
    primaryClass = "gr-qc",
    doi = "10.1103/PhysRevD.108.106022",
    journal = "Phys. Rev. D",
    volume = "108",
    number = "10",
    pages = "106022",
    year = "2023"
}

@article{Klinger:2023qna,
    author = "Klinger, Marc S. and Leigh, Robert G. and Pai, Pin-Chun",
    title = "{Extended phase space in general gauge theories}",
    eprint = "2303.06786",
    archivePrefix = "arXiv",
    primaryClass = "hep-th",
    doi = "10.1016/j.nuclphysb.2023.116404",
    journal = "Nucl. Phys. B",
    volume = "998",
    pages = "116404",
    year = "2024"
}

@article{FrancoisAndre:2023jmj,
    author = "Francois Andre, Jordan T.",
    title = "{The dressing field method for diffeomorphisms: a relational framework}",
    eprint = "2310.14472",
    archivePrefix = "arXiv",
    primaryClass = "math-ph",
    doi = "10.1088/1751-8121/ad5cad",
    journal = "J. Phys. A",
    volume = "57",
    number = "30",
    pages = "305203",
    year = "2024"
}

@article{Baulieu:2024oql,
    author = "Baulieu, Laurent and Wetzstein, Tom",
    title = "{BRST covariant phase space and holographic Ward identities}",
    eprint = "2405.18898",
    archivePrefix = "arXiv",
    primaryClass = "hep-th",
    doi = "10.1007/JHEP10(2024)055",
    journal = "JHEP",
    volume = "10",
    pages = "055",
    year = "2024"
}

@article{Adami:2024gdx,
    author = "Adami, H. and Golshani, M. and Sheikh-Jabbari, M. M. and Taghiloo, V. and Vahidinia, M. H.",
    title = "{Covariant phase space formalism for fluctuating boundaries}",
    eprint = "2407.03259",
    archivePrefix = "arXiv",
    primaryClass = "hep-th",
    doi = "10.1007/JHEP09(2024)157",
    journal = "JHEP",
    volume = "09",
    pages = "157",
    year = "2024"
}

@article{Langenscheidt:2024nyw,
    author = "Langenscheidt, Simon and Oriti, Daniele",
    title = "{New edge modes and corner charges for first-order symmetries of 4D gravity}",
    eprint = "2408.01809",
    archivePrefix = "arXiv",
    primaryClass = "gr-qc",
    doi = "10.1088/1361-6382/adbfee",
    journal = "Class. Quant. Grav.",
    volume = "42",
    number = "7",
    pages = "075010",
    year = "2025"
}

@article{Baulieu:2024rfp,
    author = "Baulieu, Laurent and Wetzstein, Tom and Wu, Siye",
    title = "{BRST Noether theorem and corner charge bracket}",
    eprint = "2411.17829",
    archivePrefix = "arXiv",
    primaryClass = "hep-th",
    doi = "10.1007/JHEP03(2026)059",
    journal = "JHEP",
    volume = "03",
    pages = "059",
    year = "2026"
}

@article{Neri:2025fsh,
    author = "Neri, Giulio and Varrin, Ludovic",
    title = "{Orbit method for quantum corner symmetries}",
    eprint = "2507.10683",
    archivePrefix = "arXiv",
    primaryClass = "hep-th",
    doi = "10.1103/hljn-f9nc",
    journal = "Phys. Rev. D",
    volume = "112",
    number = "10",
    pages = "104010",
    year = "2025"
}

@article{Golshani:2026lvc,
    author = "Golshani, M. and Sheikh-Jabbari, M. M. and Taghiloo, V. and Vahidinia, M. H.",
    title = "{Geometric aspects of Covariant Phase Space Formalism: Solution Space Slicings and surface charge integrability}",
    eprint = "2603.14439",
    archivePrefix = "arXiv",
    primaryClass = "hep-th",
    doi = "10.1007/JHEP08(2026)099",
    journal = "JHEP",
    volume = "08",
    pages = "099",
    year = "2026"
}

@article{Fefferman:2007rka,
    author = "Fefferman, Charles and Graham, C. Robin",
    title = "{The ambient metric}",
    eprint = "0710.0919",
    archivePrefix = "arXiv",
    primaryClass = "math.DG",
    journal = "Ann. Math. Stud.",
    volume = "178",
    pages = "1--128",
    year = "2011"
}

@article{Grumiller:2017sjh,
    author = "Grumiller, Daniel and Merbis, Wout and Riegler, Max",
    title = "{Most general flat space boundary conditions in three-dimensional Einstein gravity}",
    eprint = "1704.07419",
    archivePrefix = "arXiv",
    primaryClass = "hep-th",
    reportNumber = "TUW-17-04",
    doi = "10.1088/1361-6382/aa8004",
    journal = "Class. Quant. Grav.",
    volume = "34",
    number = "18",
    pages = "184001",
    year = "2017"
}

@article{Grumiller:2016pqb,
    author = "Grumiller, Daniel and Riegler, Max",
    title = "{Most general AdS$_{3}$ boundary conditions}",
    eprint = "1608.01308",
    archivePrefix = "arXiv",
    primaryClass = "hep-th",
    reportNumber = "TUW-16-16",
    doi = "10.1007/JHEP10(2016)023",
    journal = "JHEP",
    volume = "10",
    pages = "023",
    year = "2016"
}

@article{Rignon-Bret:2024wlu,
    author = "Rignon-Bret, Antoine and Speziale, Simone",
    title = "{General covariance and boundary symmetry algebras}",
    eprint = "2403.00730",
    archivePrefix = "arXiv",
    primaryClass = "hep-th",
    doi = "10.1209/0295-5075/adc2cc",
    journal = "EPL",
    volume = "149",
    number = "6",
    pages = "69002",
    year = "2025"
}

@article{Odak:2022ndm,
    author = "Odak, Gloria and Rignon-Bret, Antoine and Speziale, Simone",
    title = "{Wald-Zoupas prescription with soft anomalies}",
    eprint = "2212.07947",
    archivePrefix = "arXiv",
    primaryClass = "hep-th",
    doi = "10.1103/PhysRevD.107.084028",
    journal = "Phys. Rev. D",
    volume = "107",
    number = "8",
    pages = "084028",
    year = "2023"
}

@article{Geiller:2019bti,
    author = "Geiller, Marc and Jai-akson, Puttarak",
    title = "{Extended actions, dynamics of edge modes, and entanglement entropy}",
    eprint = "1912.06025",
    archivePrefix = "arXiv",
    primaryClass = "hep-th",
    doi = "10.1007/JHEP09(2020)134",
    journal = "JHEP",
    volume = "09",
    pages = "134",
    year = "2020"
}

@article{Geiller:2017whh,
    author = "Geiller, Marc",
    title = "{Lorentz-diffeomorphism edge modes in 3d gravity}",
    eprint = "1712.05269",
    archivePrefix = "arXiv",
    primaryClass = "gr-qc",
    doi = "10.1007/JHEP02(2018)029",
    journal = "JHEP",
    volume = "02",
    pages = "029",
    year = "2018"
}

@article{Geiller:2017xad,
    author = "Geiller, Marc",
    title = "{Edge modes and corner ambiguities in 3d Chern{\textendash}Simons theory and gravity}",
    eprint = "1703.04748",
    archivePrefix = "arXiv",
    primaryClass = "gr-qc",
    doi = "10.1016/j.nuclphysb.2017.09.010",
    journal = "Nucl. Phys. B",
    volume = "924",
    pages = "312--365",
    year = "2017"
}

@article{Ruzziconi:2020wrb,
    author = "Ruzziconi, Romain and Zwikel, C{\'e}line",
    title = "{Conservation and Integrability in Lower-Dimensional Gravity}",
    eprint = "2012.03961",
    archivePrefix = "arXiv",
    primaryClass = "hep-th",
    doi = "10.1007/JHEP04(2021)034",
    journal = "JHEP",
    volume = "04",
    pages = "034",
    year = "2021"
}

@article{Freidel:2020ayo,
    author = "Freidel, Laurent and Geiller, Marc and Pranzetti, Daniele",
    title = "{Edge modes of gravity. Part III. Corner simplicity constraints}",
    eprint = "2007.12635",
    archivePrefix = "arXiv",
    primaryClass = "hep-th",
    doi = "10.1007/JHEP01(2021)100",
    journal = "JHEP",
    volume = "01",
    pages = "100",
    year = "2021"
}

@article{Ciambelli:2022cfr,
    author = "Ciambelli, Luca and Leigh, Robert G.",
    title = "{Universal corner symmetry and the orbit method for gravity}",
    eprint = "2207.06441",
    archivePrefix = "arXiv",
    primaryClass = "hep-th",
    doi = "10.1016/j.nuclphysb.2022.116053",
    journal = "Nucl. Phys. B",
    volume = "986",
    pages = "116053",
    year = "2023"
}

@article{Ciambelli:2021nmv,
    author = "Ciambelli, Luca and Leigh, Robert G. and Pai, Pin-Chun",
    title = "{Embeddings and Integrable Charges for Extended Corner Symmetry}",
    eprint = "2111.13181",
    archivePrefix = "arXiv",
    primaryClass = "hep-th",
    doi = "10.1103/PhysRevLett.128.171302",
    journal = "Phys. Rev. Lett.",
    volume = "128",
    year = "2022"
}

@article{Ashtekar:1990gc,
    author = "Ashtekar, Abhay and Bombelli, Luca and Reula, Oscar",
    title = "{THE COVARIANT PHASE SPACE OF ASYMPTOTICALLY FLAT GRAVITATIONAL FIELDS}",
    reportNumber = "PRINT-90-0318 (SYRACUSE)",
    month = "5",
    year = "1990"
}

@article{Barnich:2009se,
    author = "Barnich, Glenn and Troessaert, Cedric",
    title = "{Symmetries of asymptotically flat 4 dimensional spacetimes at null infinity revisited}",
    eprint = "0909.2617",
    archivePrefix = "arXiv",
    primaryClass = "gr-qc",
    reportNumber = "ULB-TH-09-24",
    doi = "10.1103/PhysRevLett.105.111103",
    journal = "Phys. Rev. Lett.",
    volume = "105",
    pages = "111103",
    year = "2010"
}

@article{Freidel:2020svx,
    author = "Freidel, Laurent and Geiller, Marc and Pranzetti, Daniele",
    title = "{Edge modes of gravity. Part II. Corner metric and Lorentz charges}",
    eprint = "2007.03563",
    archivePrefix = "arXiv",
    primaryClass = "hep-th",
    doi = "10.1007/JHEP11(2020)027",
    journal = "JHEP",
    volume = "11",
    pages = "027",
    year = "2020"
}

@article{Wald:1999wa,
    author = "Wald, Robert M. and Zoupas, Andreas",
    title = "{A General definition of 'conserved quantities' in general relativity and other theories of gravity}",
    eprint = "gr-qc/9911095",
    archivePrefix = "arXiv",
    doi = "10.1103/PhysRevD.61.084027",
    journal = "Phys. Rev. D",
    volume = "61",
    pages = "084027",
    year = "2000"
}

@article{Alessio:2020ioh,
    author = "Alessio, Francesco and Barnich, Glenn and Ciambelli, Luca and Mao, Pujian and Ruzziconi, Romain",
    title = "{Weyl charges in asymptotically locally AdS$_3$ spacetimes}",
    eprint = "2010.15452",
    archivePrefix = "arXiv",
    primaryClass = "hep-th",
    reportNumber = "CJQS-2021-002",
    doi = "10.1103/PhysRevD.103.046003",
    journal = "Phys. Rev. D",
    volume = "103",
    number = "4",
    pages = "046003",
    year = "2021"
}

@article{Compere:2013bya,
    author = "Comp\`ere, Geoffrey and Song, Wei and Strominger, Andrew",
    title = "{New Boundary Conditions for AdS3}",
    eprint = "1303.2662",
    archivePrefix = "arXiv",
    primaryClass = "hep-th",
    doi = "10.1007/JHEP05(2013)152",
    journal = "JHEP",
    volume = "05",
    pages = "152",
    year = "2013"
}

@article{Troessaert:2013fma,
    author = "Troessaert, C\'edric",
    title = "{Enhanced asymptotic symmetry algebra of $AdS_{3}$}",
    eprint = "1303.3296",
    archivePrefix = "arXiv",
    primaryClass = "hep-th",
    doi = "10.1007/JHEP08(2013)044",
    journal = "JHEP",
    volume = "08",
    pages = "044",
    year = "2013"
}

@article{Speranza:2017gxd,
    author = "Speranza, Antony J.",
    title = "{Local phase space and edge modes for diffeomorphism-invariant theories}",
    eprint = "1706.05061",
    archivePrefix = "arXiv",
    primaryClass = "hep-th",
    doi = "10.1007/JHEP02(2018)021",
    journal = "JHEP",
    volume = "02",
    pages = "021",
    year = "2018"
}

@article{Barnich:2011mi,
    author = "Barnich, Glenn and Troessaert, Cedric",
    title = "{BMS charge algebra}",
    eprint = "1106.0213",
    archivePrefix = "arXiv",
    primaryClass = "hep-th",
    reportNumber = "ULB-TH-11-10",
    doi = "10.1007/JHEP12(2011)105",
    journal = "JHEP",
    volume = "12",
    pages = "105",
    year = "2011"
}

@article{Skenderis:2000in,
    author = "Skenderis, Kostas",
    editor = "Semikhatov, A. and Vasilev, M. and Zaikin, V.",
    title = "{Asymptotically Anti-de Sitter space-times and their stress energy tensor}",
    eprint = "hep-th/0010138",
    archivePrefix = "arXiv",
    reportNumber = "PUPT-1963",
    doi = "10.1142/S0217751X0100386X",
    journal = "Int. J. Mod. Phys. A",
    volume = "16",
    pages = "740--749",
    year = "2001"
}

@article{Skenderis:1999nb,
    author = "Skenderis, Kostas and Solodukhin, Sergey N.",
    title = "{Quantum effective action from the AdS / CFT correspondence}",
    eprint = "hep-th/9910023",
    archivePrefix = "arXiv",
    reportNumber = "SPIN-1999-23",
    doi = "10.1016/S0370-2693(99)01467-7",
    journal = "Phys. Lett. B",
    volume = "472",
    pages = "316--322",
    year = "2000"
}

@article{Fiorucci:2020xto,
    author = "Fiorucci, Adrien and Ruzziconi, Romain",
    title = "{Charge algebra in Al(A)dS$_{n}$ spacetimes}",
    eprint = "2011.02002",
    archivePrefix = "arXiv",
    primaryClass = "hep-th",
    doi = "10.1007/JHEP05(2021)210",
    journal = "JHEP",
    volume = "05",
    pages = "210",
    year = "2021"
}

@article{Balasubramanian:1999re,
    author = "Balasubramanian, Vijay and Kraus, Per",
    title = "{A Stress tensor for Anti-de Sitter gravity}",
    eprint = "hep-th/9902121",
    archivePrefix = "arXiv",
    reportNumber = "HUTP-99-A002, EFI-99-6, NSF-ITP-98-132",
    doi = "10.1007/s002200050764",
    journal = "Commun. Math. Phys.",
    volume = "208",
    pages = "413--428",
    year = "1999"
}

@article{Gibbons:1976ue,
    author = "Gibbons, G. W. and Hawking, S. W.",
    title = "{Action Integrals and Partition Functions in Quantum Gravity}",
    reportNumber = "PRINT-76-0995 (CAMBRIDGE)",
    doi = "10.1103/PhysRevD.15.2752",
    journal = "Phys. Rev. D",
    volume = "15",
    pages = "2752--2756",
    year = "1977"
}

@article{Brown:1992br,
    author = "Brown, J. David and York, Jr., James W.",
    title = "{Quasilocal energy and conserved charges derived from the gravitational action}",
    eprint = "gr-qc/9209012",
    archivePrefix = "arXiv",
    reportNumber = "IFP-423-UNC, TAR-009-UNC",
    doi = "10.1103/PhysRevD.47.1407",
    journal = "Phys. Rev. D",
    volume = "47",
    pages = "1407--1419",
    year = "1993"
}

@article{deHaro:2000vlm,
    author = "de Haro, Sebastian and Solodukhin, Sergey N. and Skenderis, Kostas",
    title = "{Holographic reconstruction of space-time and renormalization in the AdS / CFT correspondence}",
    eprint = "hep-th/0002230",
    archivePrefix = "arXiv",
    reportNumber = "SPIN-2000-05, ITP-UU-00-03, PUTP-1921",
    doi = "10.1007/s002200100381",
    journal = "Commun. Math. Phys.",
    volume = "217",
    pages = "595--622",
    year = "2001"
}

\end{document}